\documentclass[12pt,a4paper,onecolumn,accepted=2024-06-21]{quantumarticle}
\pdfoutput=1
\usepackage{amsfonts}
\usepackage[countmax]{subfloat}
\usepackage{algorithm, algorithmicx}
\usepackage[noend]{algpseudocode}
\usepackage{graphics}
\usepackage{tikz}
\usetikzlibrary{decorations.markings,quantikz}
\usepackage{enumerate}
\usepackage{hyperref}
\usepackage{ulem}
\usepackage{mathtools}
\mathtoolsset{showonlyrefs}
\usepackage{amsmath}
\usepackage{graphicx}
\usepackage{enumerate}
\usepackage{amssymb}
\usepackage{color}
\usepackage{amsthm}
\usepackage{caption}
\usepackage{subcaption}
\usepackage{cite}

\usepackage{caption}
\usepackage{subcaption}

\usepackage[noend]{algpseudocode}

\usepackage[framemethod=tikz]{mdframed}
\usepackage{comment}
\usepackage{tikz}
\usepackage{doi}
\usepackage{url}

\usepackage[left=0.75in,right=0.75in,top=0.75in,bottom=0.75in]{geometry}
\newcommand{\appref}[1]{\hyperref[#1]{{Appendix~\ref*{#1}}}}
\newcommand{\be}{\begin{eqnarray} \begin{aligned}}
\newcommand{\ee}{\end{aligned} \end{eqnarray} }
\newcommand{\benn}{\begin{eqnarray*} \begin{aligned}}
\newcommand{\eenn}{\end{aligned} \end{eqnarray*}}

\newcommand*{\cA}{\mathcal{A}} 
\newcommand*{\cB}{\mathcal{B}}
\newcommand*{\cC}{\mathcal{C}}
\newcommand*{\cE}{\mathcal{E}}

\newcommand*{\cG}{\mathcal{G}}

\newcommand*{\cJ}{\mathcal{J}}

\newcommand*{\cL}{\mathcal{L}}
\newcommand*{\cM}{\mathcal{M}}
\newcommand*{\cN}{\mathcal{N}}
\newcommand*{\cD}{\mathcal{D}}
\newcommand*{\cQ}{\mathcal{Q}}
\newcommand*{\cR}{\mathcal{R}}
\newcommand*{\cO}{\mathcal{O}}
\newcommand*{\cP}{\mathcal{P}}

\newcommand*{\tr}{\mathop{\mathrm{tr}}\nolimits}

\newcommand{\bc}{\begin{center}}
\newcommand{\ec}{\end{center}}

\newtheorem{theorem}{Theorem}[section]
\newtheorem{lemma}[theorem]{Lemma}

\usepackage{amsfonts}

\def\01{\{0,1\}}

\newcommand*{\idG}{\mathsf{e}}

\begin{document}

\newcommand*{\Parity}{\mathsf{PARITY}}
\newcommand*{\AC}{\mathsf{AC}}
\newcommand*{\NC}{\mathsf{NC}}
\newcommand*{\QNC}{\mathsf{QNC}}
\newcommand*{\BP}{\mathsf{BP}}
\renewcommand*{\P}{\mathsf{P}}
\renewcommand*{\L}{\mathsf{L}}
\newcommand*{\TC}{\mathsf{TC}}
\newcommand*{\NL}{\mathsf{NL}}

\newcommand*{\DT}{\mathsf{DT}}
\newcommand*{\poly}{\mathsf{poly}}
\newcommand*{\rpoly}{\mathsf{rpoly}}

\newcommand*{\twoPCliff}[1]{\mathsf{2CliffP}(#1)}

\newcommand*{\WP}{\mathsf{WP}}

\newcommand{\ZZ}{\mathbb{Z}}

\newcommand{\labelgroup}[5]{\POS"#1,#5"."#2,#5"."#1,#5"."#2,#5", \POS"#1,#5"."#2,#5"."#1,#5"."#2,#5"*!C!<1em,#3>=<0em>{#4}}
\newcommand*{\cfinal}{{\prod_{i=1}^L C_i}}

\DeclareDocumentCommand{\makereversebit}{O{+45}O{}m}{
	\arrow[arrows,line cap=round,to path={(\tikztostart) -- ($(\tikztostart)!{+0.5/cos(#1)}!#1:(\tikztotarget)$) node [anchor=west,style={#2}]{#3} -- (\tikztotarget)}]{d}
}

\newcommand*{\teleport}{\mathsf{Telep}}
\newcommand*{\psim}[1]{\mathsf{PSIM}\left[#1\right]}

\newcommand*{\parityproblem}{\mathsf{PARITY}}
\newcommand*{\PCliff}[1]{\mathsf{CliffP}(#1)}
\newcommand*{\solves}[2]{\left(#1 \textrm{ solves } #2\right)}

\title{Single-qubit gate teleportation provides\newline a quantum advantage}
\author{Libor Caha, Xavier Coiteux-Roy and  Robert K\"onig}
\affil{\small School of Computation, Information and Technology, Technical University of Munich \& \\
Munich Center for Quantum Science and Technology, Munich, Germany.}
\maketitle

\begin{abstract}
    Gate-teleportation circuits are arguably among the most basic examples of  computations  believed to provide a quantum computational advantage: in seminal work~\cite{TerhalDiVincenzo04}, Terhal and DiVincenzo have shown that these circuits elude simulation by efficient classical algorithms under plausible complexity-theoretic assumptions. Here we consider possibilistic  simulation~\cite{wang2021possibilistic}, a particularly weak form of this task where the goal is to output any string appearing with non-zero probability in the output distribution of the circuit.  We 
    show that even for single-qubit Clifford-gate-teleportation circuits this simulation problem cannot be solved by constant-depth classical circuits with bounded fan-in gates. Our results are unconditional and are obtained by a reduction to the problem of computing parity, a well-studied problem in classical circuit complexity.
\end{abstract}
\vspace*{-1cm}
\tableofcontents

\section{Introduction and main result}

A quantum circuit~$\cQ$ on $n$~qubits takes as input a classical bit string $I=(I_1,\ldots,I_m)\in \{0,1\}^m$. Starting from a product state~$\ket{0}^{\otimes n}$, it  applies  one- and two-qubit gates that may be classically controlled by constant-size subsets~$\{I_j\}_{j}$ of  input bits, and finally measures each qubit in the computational basis to yield an outcome~$O\in \{0,1\}^n$.  A prominent example is the single-qubit (Clifford)-gate-teleportation circuit~$\teleport_n$  (on $2n$~qubits) shown in Fig.~\ref{fig:teleportationcircuit}. The terminology used here alludes to the fact that it is composed of several copies of the gate-teleportation protocol of Gottesman and Chuang~\cite{gottesmanchuang}. The latter teleports an unknown quantum state while simultaneously applying a Clifford unitary to it.
 Here we seek to characterize 
 the computational power of the circuit~$\teleport_n$.

\begin{figure}[!h]
    \centering
    \begin{subfigure}[t]{0.37\linewidth}
        \includegraphics[scale=0.92]{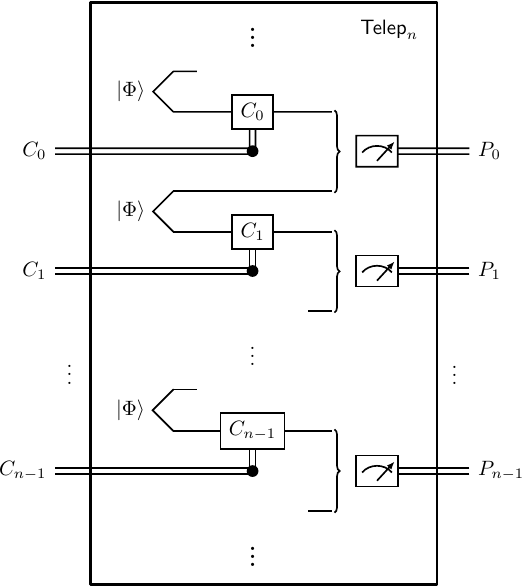}
        \caption{The~$\teleport_n$ circuit.}
        \label{fig:telep}
    \end{subfigure}
    \hfill
    \begin{subfigure}[t]{0.4\linewidth}
        \hspace*{-0.3cm}\includegraphics[width=1.1\textwidth]{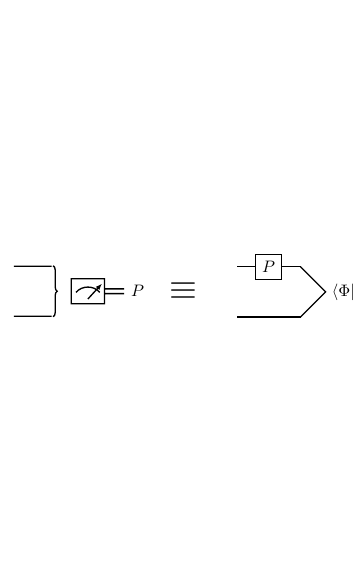}
        \caption{Pauli operator output after the Bell measurement with $\ket{\Phi}=\frac{1}{\sqrt{2}}(\ket{00}+\ket{11})$.}
        \label{fig:BellMeasurement}
    \end{subfigure}
    \caption{The gate-teleportation circuit~$\teleport_n$. The circuit takes as input an~$n$-tuple~$(C_0,\ldots,C_{n-1})\in\cC^n$ of single-qubit Clifford group elements. 
			Each Clifford is applied to one half of a Bell state. Subsequently, Bell-state measurements are performed on pairs of qubits. The output can be associated with an $n$-tuple~$(P_0,\ldots,P_{n-1})\in \cP^n$ of Pauli elements from the set~$\cP=\{I,X,Y,Z\}$ corresponding to different outcomes of Bell measurements. The circuit has a cyclic structure: the $n$-th Bell measurement is applied on the first and last qubits in the figure. As the single-qubit Clifford group~$\cC=\cP\cdot \{I,H,S,HS,SHS,(HS)^2\}$ consists of $|\cC|=24$ elements, its elements can be encoded into $5$-bit strings. Similarly, elements of~$\cP$ can be encoded as $2$-bit strings (pairs of bits). Correspondingly, we will often interpret inputs
			$(C_0,\ldots,C_{n-1})$ as elements of~$(\{0,1\}^5)^n$
			and outputs $(P_0,\ldots,P_{n-1})$ as elements of~$(\{0,1\}^2)^n$. 
            \label{fig:teleportationcircuit}
}
\end{figure}

The behavior of a quantum circuit~$\cQ$ is  fully specified by
the conditional distribution~$p^{\cQ}(O|I)$ of the output~$O$ given a fixed input~$I$. Its computational power is directly tied to the difficulty of sampling from~$p^{\cQ}(\cdot |I)$ for different inputs~$I$, and, more generally, the hardness of computing or approximating properties of this collection of distributions. Corresponding  notions of classically simulating~$\cQ$ give rise to various computational problems with more stringent notions of simulation  associated with harder computational problems. 
Arguably one of the weakest notions of simulating~$\cQ$ is  that of {\em possibilistic simulation}, a convenient term introduced by Wang in~\cite{wang2021possibilistic} for a concept implicitly used in earlier work~\cite{bravyigossetkoenig,BravyiGossetKoenigTomamichel}. The corresponding computational problem, which we denote by~$\psim{\cQ}$, asks, for a given input~$I$ (an {\em instance} of~$\psim{\cQ}$), to produce~$O\in \{0,1\}^n$ such that $p^{\cQ}(O|I)>0$, i.e., such that $O$~occurs with non-zero probability in the output distribution of~$\cQ$. 

Our main result is that, for the gate-teleportation circuit~$\teleport$, even the problem~$\psim{\teleport}$ is hard for constant-depth classical circuits. To state this, recall that an $\NC^0$-circuit is a constant-depth polynomial-size classical circuit with bounded fan-in gates (where, following standard terminology, we often refer to an infinite circuit family simply as a circuit).
\begin{theorem}\label{thm:main}
The~$\teleport$ quantum circuit trivially solves~$\psim{\teleport}$ for every instance.
In contrast, any (even non-uniform) $\NC^0$-circuit fails to solve~$\psim{\teleport}$ on certain instances. 
\end{theorem}
We note that average-case statements for suitable distributions over instances of~$\psim{\teleport}$ can be derived in a similar manner as e.g.,~\cite[Supplementary material]{bravyigossetkoenig}. Here we focus on the worst case for simplicity.

The study of the difficulty  of classically simulating quantum gate teleportation  was initiated in seminal work by Terhal and DiVincenzo~\cite{TerhalDiVincenzo04}. It was shown in~\cite{TerhalDiVincenzo04} (see also~\cite{AaronsonCh15,fennergreenhomerzhang05}) that there is no efficient classical algorithm that can sample exactly from the output distribution  of  a variant of the gate-teleportation circuit unless the polynomial hierarchy collapses. The gate-teleportation circuit used in this argument acts in parallel on polynomially many qubits, and realizes any polynomial-size quantum circuit by a post-selected constant-depth quantum circuit. Furthermore, gate-teleportation circuits when combined with post-selection gives the complexity class~$\mathsf{PostBQP}$ of (uniform) polynomial-size  post-selected quantum circuits.
This leads to the hardness of classically sampling from their output distribution (also called weak simulation in Ref.~\cite{VanDenNesClassicalSimulation}) under plausible complexity-theoretic assumptions. 

This line of reasoning has also been used  to show hardness of simulating ``commuting quantum computation" as formalized by IQP circuits~\cite{bremnerjozsashepher11}, level-$1$-quantum approximate optimization (QAOA) circuits~\cite{Farhiharrow2016} and boson sampling circuits~\cite{aaronsonarkhipovboson11}.

\subsection*{Proof outline}
We establish Theorem~\ref{thm:main}  by
connecting the gate-teleportation simulation problem $\psim{\teleport}$
to a basic (trivial) result from classical complexity theory, namely the fact that the parity function cannot be computed by an $\NC^0$-circuit.
Our argument proceeds through a reduction to an intermediate computational problem which relates $\psim{\teleport}$ and $\parityproblem$.

Let us briefly describe the intermediate problem and give a high-level overview of our reasoning.

\subsubsection*{Quantum advantage against $\NC^0$}
The intermediate problem we use in our analysis is another simulation problem $\psim{\PCliff{Q}_L}$ for a two-qubit quantum circuit~$\PCliff{Q}_L$.
The circuit~$\PCliff{Q}_L$ is  defined by using an (arbitrary but fixed) function $Q:\cC^L\rightarrow \cP$. The circuit~$\PCliff{Q}_L$ 
is illustrated in Fig.~\ref{fig:PCliff}. 
Here and below we denote by~$\cP$ the set $\cP=\{I,X,Y,Z\}$
of Pauli operators and by~$\cC$ the set~$\cC=\cP\cdot \{I,H,S,HS,SHS,(HS)^2\}$ of single-qubit Cliffords.
(These sets are representatives of different cosets of the Pauli respectively Clifford groups when taking the quotient with respect to global phases. We refer to e.g.,~\cite{Grier_2022} for details.)
It takes as input $L+1$ single-qubit Clifford gates
\begin{align}
(C_1,\ldots,C_L,D)\in \cC^{L+1}
\end{align} 
and outputs a Pauli~$P\in\cP$, see the caption of Fig.~\ref{fig:teleportationcircuit} for a definition of these sets. We call $\PCliff{Q}_L$  the {\em Pauli-corrected Clifford circuit with correction function~$Q$}. 

\begin{figure}[!h]
		\centering
		\includegraphics[scale=1.175]{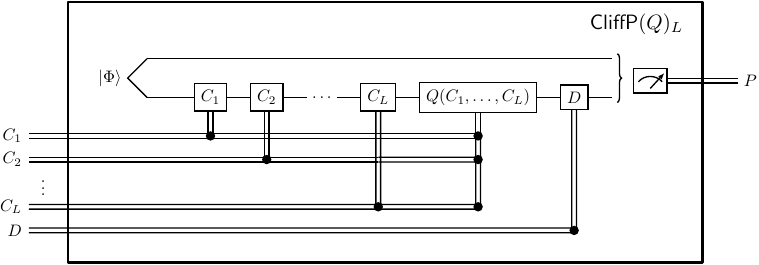}
		\caption{The Pauli-corrected Clifford circuit $\PCliff{Q}_L$ with correction function $Q$: on input~$(C_1,\ldots,C_L,D)\in\cC^{L+1}$, the
		Clifford gates~$C_1,\ldots,C_L$ are successively applied to one half of a Bell state~$\ket{\Phi}$. A Pauli ``correction'' operator~$Q(C_1,\ldots,C_L)$ is applied to the resulting state. Finally, the state is measured in a basis of maximally entangled states that is obtained by rotating the Bell basis (as specified by the Clifford~$D$). The output is a single-qubit Pauli~$P\in \cP$ obtained from a Bell measurement.		\label{fig:PCliff}
		}
\end{figure}

In Section~\ref{sec:advantageagainstNCzero}, we 
show the following: assume that~$\mathsf{C}$ is an (even non-uniform) $\NC^0$-circuit 
solving~$\psim{\teleport_n}$ (i.e., producing a valid output for all instances). Then, for $L=\Theta(n)$, there is a function~$Q:\cC^L\rightarrow\cP$ (that depends  on~$\mathsf{C}$) and a non-uniform $\NC^0$-circuit~$\mathsf{C}'$ solving~$\psim{\PCliff{Q}_L}$. The construction of~$\mathsf{C'}$  makes use of~$\mathsf{C}$ in a  non-black-box manner and relies on no-signaling arguments from Section~\ref{sec:nosignaling}. We summarize this (with a slight abuse of notation) as
\begin{align}
\!\!\!\!\!\left(\psim{\teleport}\in \NC^0/\poly\right) \qquad \Rightarrow \qquad \left(\exists Q\colon  \psim{\PCliff{Q}}\in \NC^0/\poly\right)\,,\label{eq:nczeroseparation}
\end{align}
where $\NC^0/\poly$ is the class of problems solved by non-uniform families of $\NC^0$-circuits.

Our second reduction heavily borrows from ideas developed in~\cite{grier2020interactive} in the context of interactive quantum advantage protocols. By adapting (specializing) these techniques, we show that the ability of possibilistically simulating~$\PCliff{Q}_L$ implies the ability to compute products of single-qubit Clifford gates by a randomized $\NC^0$-circuit with success probability at least $2/3$. In particular, given a (deterministic) oracle~$\cO$ for~$\psim{\PCliff{Q}}$ (for any fixed function~$Q$),  this immediately yields an  $\NC^0$-circuit
with additional uniformly random input bits and oracle access to~$\cO$, i.e., an ~$(\NC^0/\rpoly)^\cO$-circuit, 
which solves~$\parityproblem$ with success probability at least $2/3$. Succinctly, we summarize this as
\begin{align}
\!\!\!\!\!\!\solves{\cO}{\psim{\PCliff{Q}}} \Rightarrow\left(\!(\NC^0/\rpoly)^{\cO} \text{ solves } \parityproblem \text{ with prob.}\ge \frac23\right)\!.\label{eq:bpnczeroseparation}
\end{align}
Combining~\eqref{eq:nczeroseparation}
and~\eqref{eq:bpnczeroseparation} and
using a standard derandomization argument shows that
\begin{align}
    \left(\psim{\teleport}\in \NC^0/\poly\right) \qquad \Rightarrow \qquad 
    \left(\parityproblem\in \NC^0/\poly\right)\ . 
\end{align}
Thus the existence of such a circuit would contradict the folklore (trivial) result that~$\parityproblem$ cannot be solved by a non-uniform $\NC^0$-circuit. Indeed, an output bit of an $\NC^0$-circuit is influenced by at most a constant number of bits of an input instance, whereas  $\parityproblem$ depends on all input bits. In conclusion, this shows that there is no $\NC^0$-circuit for the teleportation simulation problem~$\psim{\teleport}$.

\subsection*{Related work}

Let us briefly summarize related work. The first unconditional separation of con-stant-depth quantum circuits and constant-depth classical circuits with bounded fan-in gates, i.e., $\QNC^0$ and $\NC^0$, was achieved by Bravyi, Gosset, and K\"onig~\cite{bravyigossetkoenig} for the \textit{$2D$ Hidden Linear Function problem}. This result was later extended to noisy quantum circuits by Bravyi, Gosset, K\"onig, and Tomamichel~\cite{BravyiGossetKoenigTomamichel}. The latter  construction uses a simpler underlying problem---the \textit{$1D$ Magic Square Game}. Both of these constructions are physically motivated and based
on violations of Bell inequalities. In subsequent work Bene Watts, Kothari, Schaeffer, and Tal~\cite{WattsKothariSchaefferTalAC0} achieved an even stronger separation, establishing a quantum advantage of constant-depth quantum circuits against $\AC^0$-circuits, i.e., the family of classical constant-depth polynomial-size circuits obtained by augmenting $\NC^0$-circuits with unbounded fan-in AND and OR gates. The problem  used in~\cite{WattsKothariSchaefferTalAC0} is the \textit{Relaxed Parity Halving problem}. Similarly to our result, this work establishes a quantum advantage based on classical circuit complexity lower bounds as opposed to the first two~\cite{bravyigossetkoenig,BravyiGossetKoenigTomamichel} Bell-inequality-based  constructions.

In this work our aim is to give a quantum advantage proposal separating $\QNC^0$ from $\NC^0$. While the mentioned prior works have already established this complexity-theoretic separation, our focus is on the minimal quantum information-processing resources required to do so. We also use a novel and potentially more generally applicable proof technique.  
To compare our circuits to these priors works,  let us  discuss what geometry and gates these constructions use.  The quantum circuit of~\cite{bravyigossetkoenig} for the $2D$ Hidden Linear Function problem, and
the circuit of~\cite{WattsKothariSchaefferTalAC0} for the
Relaxed Parity Halving problem are geometrically $2D$-local circuits~\footnote{A circuit is said to be \emph{geometrically $2D$-local} (resp. $1D$-local) if its underlying interaction graph is a $2D$ lattice (resp. $1D$ lattice).} composed of classically controlled Clifford gates.  The $1D$~Magic Square Game problem considered in~\cite{BravyiGossetKoenigTomamichel} is a classically controlled Clifford circuit on a chain of pairs of qubits.  
(This geometric locality allows for a fault-tolerant implementation by a  local quantum circuit in a $3D$~geometry.) Our proposal is 
quite similar to that of~\cite{BravyiGossetKoenigTomamichel}, but is more directly associated with a  ring of qubits rather than linear array of pairs of qubits. We note that the required resources are essentially the same despite this apparent difference. However, the use of periodic boundary conditions makes a novel argument possible, and leads to perhaps the most simple quantum circuit separating $\QNC^0$ from $\NC^0$.

Finally, we note that the work of Grier and Schaeffer~\cite{grier2020interactive} provides several quantum advantage protocols involving interaction. 
Unlike for typical computational problems such as those considered here (where the task is to provide a certain output given an input), corresponding advantage proposals involve a task which requires  several rounds of ``challenges'' and ``responses'', respectively. These can be seen as messages passed between a prover (who aims to demonstrate quantum computational power) and a verifier.
In the work~\cite{grier2020interactive}, the  quantum player (prover) can be realized by a constant-depth (classically controlled) Clifford circuit which is realized on a chain consisting of  qubits, a pair of qubits, or a grid 
depending on the complexity-theoretic result considered: for example, a chain of qubits is used to show hardness for $\AC^0$ with  additional $\textrm{MOD}_6$ gates, and chain of pairs of qubits is used to show hardness for classical circuits with bounded fan-in gates and logarithmic depth, i.e.,~$\NC^1$. We note that, while establishing several containments of classical and quantum complexity classes, this does not directly provide a separation between constant-depth quantum and constant-depth classical circuits (in the sense of separating~$\QNC^0$ from e.g.,~ $\NC^0$) because  interactive protocols are studied (i.e., these results separate suitably defined complexity classes involving interaction.)
We emphasize that, although their focus is different from ours, some 
of the key ideas developed in~\cite{grier2020interactive} are central to our derivation. 

Let us briefly comment on the relation of our work to~\cite{grier2020interactive}, where an oracle~$\tilde{R}$ for the simulation of a certain interactive process is considered. 
The proof technique in~\cite{grier2020interactive} assumes
 that the oracle~$\tilde{R}$ provides \textit{rewind} access. It uses a reduction that does not treat $\tilde{R}$ as a black box. Instead, the classical algorithm is allowed to ``rewind'' the classical simulator (implementing the oracle) to an earlier state of the computation, change inputs or internal variables and continue the computation. Notice that this is a natural requirement for a classical machine (e.g. for $\NC^0, \AC^0$ or Turing machines it is implemented by the same class) but is generally impossible in a quantum computational model involving non-unitary dynamics such as measurements.

In more detail, the simulated interactive process considered in~\cite{grier2020interactive} is conceptually similar to a Bell experiment:  in the first stage, given an $L$-tuple~$(C_1,\ldots,C_L)$ of Clifford group elements, the stabilizer state~$P(C_1,\ldots,C_L,m)(C_L\cdots C_1)\ket{\Psi_0}$ is prepared by a constant-depth quantum circuit using measurements, where the Pauli ``correction''~$P(C_1,\ldots,C_L,m)$ is determined by the Cliffords~$(C_1,\ldots,C_L)$ and the measurement results~$m$. This first stage  corresponds to preparing an initial entangled state for a Bell experiment. In the second stage of the process, an input is provided according to which a suitable Pauli measurement is applied to the prepared state. Without loss of generality, the input takes the form of a Clifford unitary~$D$, and the executed measurement is the von Neumann (i.e., projective) measurement using the basis~$\{D\ket{z}\}_{z}$. This second step corresponds to the execution of a Bell measurement on the prepared state.  The work~\cite{grier2020interactive} considers simulation problems associated with different interactive protocols of this kind for $k$-qubit Clifford circuits, embedded into so-called graph state simulation problems.

Rewind access to an oracle~$\tilde{R}$ simulating such an interactive process  means that oracle queries to~$\tilde{R}$ can be repeated in such a way that it yields an identical output~$m$ for the first step. In other words, the classical machine can be rewound to the beginning of the second stage. In contrast, reproducing this behavior using a quantum protocol would require post-selecting on  fixed measurement outcomes occurring with exponentially small probability. The ability of a classical simulator to rewind is at the root of the hardness results of~\cite{grier2020interactive} concerning interactive quantum advantage. 

In contrast to~\cite{grier2020interactive}, we do not need to assume rewind access to the classical oracle~$\tilde{R}$. This is achieved by exploiting the internal structure of the classical simulator: we consider oracles realized by circuits and provide a non-black-box reduction. As we discuss below (see Section~\ref{sec:possibilistictomographyredc}), however, our arguments  heavily rely on some of the techniques of~\cite{grier2020interactive} developed in the interactive scenario.

A more significant distinction between
our work and that of~\cite{grier2020interactive} is the fact that the computational problem we consider does not require interaction, i.e., two distinct stages. To illustrate this difference, consider the following simple example of a two-stage (interactive) quantum advantage in a problem involving two players Alice and Bob. The problem is commonly known as the CHSH game~\cite{CHSH}. Let us assume that Alice and Bob start out with no classical or quantum correlations. 
In stage~$1$, Alice and Bob are allowed to communicate arbitrarily. In stage~$2$, Alice and Bob are not allowed to communicate anymore, and they receive random bits~$x$ and $y$, respectively. They are asked to output bits~$a$ respectively $b$ such that $a\cdot b=x+y\pmod 2$. It is well-known that Alice and Bob can win with higher success probability than two classical players, if they can exchange quantum messages at stage~$1$. It is easy to see that a similarly simple quantum advantage cannot be observed when there is only a single stage (i.e., where Alice and Bob obtain~$x$ and $y$ at the very beginning). This example indicates that the challenge of establishing a quantum advantage heavily depends on
how many (interactive) rounds are considered. While our work removes the need for interaction, an even more extreme example is the problem of sampling: Here a quantum advantage has recently been established for problems ``without input''~\cite{BeneWattsParham2023unconditional} (see also~\cite{renou2019genuine}).

\section*{Outlook}
Due to the simplicity of the circuit~$\teleport$, our work  provides a starting point for a fault-tolerant $3D$-local quantum advantage scheme based on the techniques of~\cite{BravyiGossetKoenigTomamichel}.
A natural open question is whether the $\teleport$ circuit also provides a quantum advantage over $\AC^0$, with a similar separation as the (up to the best of our knowledge) strongest known unconditional result~\cite{WattsKothariSchaefferTalAC0}. More generally, it is unclear whether or not 
 quantum advantage schemes against~$\AC^0$ are 
 feasible based on geometrically $1D$-local quantum circuits only.

\section{Locality properties of classical circuits}\label{sec:nosignaling}

In this section we identify certain no-signaling properties every $\NC^0$-circuit necessarily has. These are crucial in our non-black-box reduction from the simulation problem $\psim{\PCliff{Q}}$  to the problem~$\psim{\teleport}$  in Section~\ref{sec:advantageagainstNCzero}.

In Sections~\ref{sec:blockstructurefunctionslocality} and~\ref{sec:localityofcircuits} we define lightcones and locality of functions and $\NC^0$-circuits. In Section~\ref{sec:nosignalingnc0} we use combinatorial arguments to prove no-signaling properties of $\NC^0$.

\subsection{Locality of functions \label{sec:blockstructurefunctionslocality}}
In the following, we  consider functions
\begin{align}
\begin{matrix}
    f:&(\{0,1\}^k)^n&\rightarrow& (\{0,1\}^r)^s\\
    &I &\mapsto& f(I)=(f_1(I),\ldots,f_s(I))
    \end{matrix}\label{eq:fknrsdef}
\end{align}
taking an ``input''-$n$-tuple~$I$  of $k$-bit strings to an ``output''-$s$-tuple~$f(I)$ of $r$-bit strings.
We often write~$I=(I_1,\ldots,I_n)$ for the input variables and~$O=(O_1,\ldots,O_s)$ for the output variables, where $O_m=f_m(I)$ for $m\in \{1,\ldots,s\}$. We call each of the $k$-bit-strings making up
an input~$I$ a {\em block}. That is, the~$j$-th block of the input consists of the $k$~binary variables~$I_j=(I_j^{(1)},\ldots,I_j^{(k)})$ of the argument, where $j\in \{1,\ldots,n\}$. Similarly, the $m$-th block of the output is associated with~$O_m=(O_m^{(1)},\ldots,O_m^{(r)})$ for $m\in \{1,\ldots,s\}$.

While a function of the form~\eqref{eq:fknrsdef} can be regarded as a function
\begin{align}
    f:\{0,1\}^N\rightarrow\{0,1\}^M\label{eq:fnm}
\end{align} with~$N=k\cdot n$ and~$M=r\cdot s$, i.e., as a function taking~$N$-bit strings to~$M$-bit strings, the block-structure~\eqref{eq:fknrsdef} will also be important in the following. To distinguish between the two interpretations of a function~$f$ of the form~\eqref{eq:fknrsdef}, we refer to~\eqref{eq:fnm} as the function~$f$ {\em with the block-structure ignored}. Note also that for $k=r=1$, there is no difference between these two notions. In particular, the following definitions apply to any function~$f:\{0,1\}^N\rightarrow\{0,1\}^M$ taking $N$-bit strings to~$M$-bit strings, for any $N,M\in\mathbb{N}$.

For a function~$f$ of the form~\eqref{eq:fknrsdef} we say that the $k$-th output block~$O_k$ is \emph{influenced} by the~$j$-th input~$I_j$ if and only if 
the output block~$O_k$ depends non-trivially on the value of the $j$-th input block, i.e., if there is a choice of
~$(I_1,\ldots,I_{j-1},I_{j+1},\ldots,I_n)\in (\{0,1\}^{k})^{n-1}$ and $J,J'\in \{0,1\}^k$ such that
\begin{align}
    f_k(I_1,\ldots,I_{j-1},J,I_{j+1},\ldots,I_n)&\neq f_k(I_1,\ldots,I_{j-1},J',I_{j+1},\ldots,I_n)\ .
\end{align}
In the following, we will often identify input- and output-blocks with their indices. Correspondingly, we define the {\em forward lightcone}~$\cL^\rightarrow (j)$ of the $j$-th input block~$I_j$ as the set of all output blocks that are influenced by it, i.e.,
\begin{align}
    \cL^\rightarrow (j) := \{k\in[s] \mid O_k~\textrm{is influenced by~}I_j \}\ .
\end{align}
Similarly, the backward lightcone~$ \cL^\leftarrow (k)$ of the $k$-th output block~$O_k$ is the set of all input blocks that influence it, 
\begin{align}
    \cL^\leftarrow (k) := \{j\in[n] \mid O_k~\textrm{is influenced by~}I_j \}\ . 
\end{align}
A function~$f$ is said to be {\em $\ell$-local} if the size of all  backward lightcones is bounded by $\ell$, i.e., if 
\begin{align}
| \cL^\leftarrow (k)|\le \ell\qquad\textrm{ for all }\qquad k  \in [s]\ .  
\end{align}
Lightcones naturally extend to a subset~$\cJ$ of either input and output blocks: for $\cJ_{\textrm{in}}\subseteq [n]$ and $\cJ_{\textrm{out}}\subseteq [s]$, we set
\begin{align}
    \cL^\rightarrow (\cJ_{\textrm{in}})=\bigcup_{j\in\cJ_{\textrm{in}}}\cL^\rightarrow (j) , \qquad \text{ and }\qquad
    \cL^\leftarrow (\cJ_{\textrm{out}})=\bigcup_{j\in\cJ_{\textrm{out}}} \cL^\leftarrow (j)\ .
\end{align}

For a subset~$\{I_j\}_{j\in\cJ}$
of input blocks (specified by a subset~$\cJ=\{j_1<\cdots <j_\ell\}\subseteq [n]$), we say that an output block~$O_k$ is {\em determined by~$\{I_j\}_{j\in\cJ}$} if 
it is only a function of the corresponding input variables, i.e., if there is a function~$g:(\{0,1\}^k)^{\ell}\rightarrow \{0,1\}^r$ such that
\begin{align}
    f_k(I)&=g(I_{j_1},\ldots,I_{j_\ell})\qquad\textrm{ for all }\qquad I=(I_1,\ldots,I_n)\in (\{0,1\}^k)^n\ .
\end{align}

\subsection{Locality of circuits}\label{sec:localityofcircuits}
Note that for circuit~$\mathsf{C}\colon\{0,1\}^N\rightarrow \{0,1\}^M$ of depth $d$, the maximal fan-in~$\nu$ of the gates in the circuit sets an upper bound on the size of all backward lightcones: we have 
\begin{align}
    |\cL^\leftarrow (j)|\le \nu^d\qquad\textrm{ for all }\qquad j\in [M]\ .
\end{align}
In particular, a depth-$d$ circuit with gates of fan-in upper bounded by~$\nu$ is $\ell$-local with 
\begin{align}
    \ell \leq \nu^d\ .\label{eq:nczerocircuitbounddepth}
\end{align}

Similar considerations
apply to every ``block'' function~$\mathsf{C}\colon(\{0,1\}^k)^n\rightarrow (\{0,1\}^r)^{s}$ that is realized by a circuit~$\mathsf{C}$. If each gate has fan-in upper bounded by~$\nu$ (i.e., has at most~$\nu$ bits as input), then each output bit of the circuit can be influenced by at most $\nu^d$~input bits. Thus each output bit can be influenced by at most~$\nu^d$ blocks. Since every input block consists of $k$~bits, it follows that
\begin{align}
    |\cL^\leftarrow (j)|\le k\nu^d\qquad\textrm{ for all }\qquad j\in [s]\ ,\label{eq:cljupperboundz}
\end{align}
i.e., the circuit is {\em$\ell$-(block-)local} with~$\ell \leq k\nu^d$. 

In the following, we often consider 
families of functions
\begin{align}
\{f_n: (\{0,1\}^k)^n\rightarrow (\{0,1\}^r)^{s(n)}\}_{n\in\mathbb{N}}
\end{align}
where $k,r$~are constant and $s(n)=\poly(n)$ is polynomial in $n$.
Eq.~\eqref{eq:cljupperboundz}
implies that if such a function family is represented by an~$\NC^0$-circuit, then it has constant locality, i.e., it is $O(1)$-local.

\subsection{No-signaling features of $\NC^0$-circuits}
\label{sec:nosignalingnc0}

We establish no-signaling properties for a specific class of $\NC^0$-circuits that have the same number of input and output blocks. These are exactly the type of circuits we want to analyze for the $\psim{\teleport}$ problem. 
Let $\mathsf{C}\colon(\{0,1\}^k)^n\rightarrow (\{0,1\}^r)^{n}$
be a function.
That is, $\mathsf{C}$ takes as input~$(I_1,\ldots,I_n)\in (\{0,1\}^k)^n$ and outputs $(O_1,\ldots, O_n)\in (\{0,1\}^r)^n$. 
We assume that the circuit $\mathsf{C}$ has (block-)locality~$\ell$.

Let $L\in [n]$. We say that $m\in [n]$ is {\em $L$-good} if and only if
\begin{align}
    \cL^\rightarrow(m) \cap [L]&=\emptyset\ .\label{eq:Lgooddef}
\end{align}
That is, an input index~$m\in [n]$ is $L$-good if and only if the input block~$I_m$ does not influence the first $L$~output blocks~$(O_1,\ldots,O_L)$. We illustrate the  $L$-good property in Figure~\ref{fig:Lgood}. Let us denote
by $\mathsf{Good}(L)\subseteq [n]$ the subset of $m\in [n]$ such that~$m$ is $L$-good. We also set $\mathsf{Bad}(L):=[n]\backslash \mathsf{Good}(L)$.
 \begin{figure}[!ht]
		\centering
		\begin{subfigure}[t]{0.45\textwidth}
			\centering
		\includegraphics[]{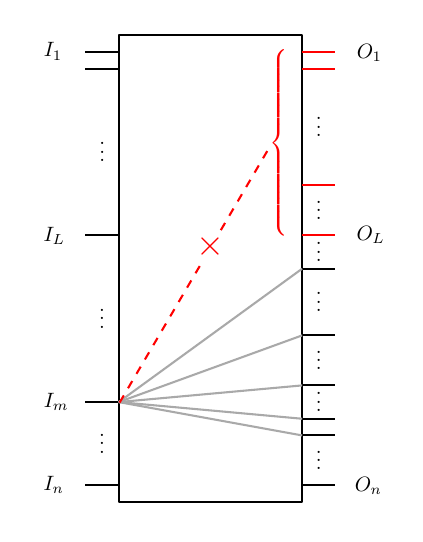}
		\caption{An illustration of an \emph{$L$-good} $m\in [n
  ]$, i.e., the forward lightcone of the input block~$I_m$ has no intersection with the outputs~$O_1, \ldots, O_L$.}\label{fig:Lgood} 
		\end{subfigure}
		\qquad
		\begin{subfigure}[t]{0.45\textwidth}
			\centering
		\includegraphics[]{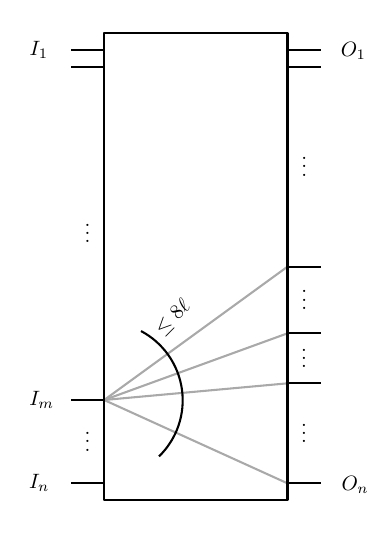}
		\caption{An illustration of an element $m\in [n]$ with \emph{limited signaling}, i.e., the size of forward lightcone of the input block $I_m$ is upper-bounded by $8\ell$.}\label{fig:LimitedSignaling}
		\end{subfigure}
		\caption{Structural properties of local functions.}
\end{figure}

In a similar fashion, let us say that the $m$-th input block~$I_m$ has \emph{limited signaling} if and only if
\begin{align}
    |\cL^{\rightarrow}(m)|\leq 8\ell \ ,
\end{align}
i.e., if and only if~$I_m$ influences at most~$8\ell$ output blocks. (The choice of factor~$8$ is for convenience only.) We depict the limited signaling property in Figure~\ref{fig:LimitedSignaling}.
Let us denote the set of limited-signaling input indices by~$\mathsf{LimitedSignaling}\subseteq [n]$.

Our goal in this section is to show that for a  suitably chosen value of~$L$, there are indices~$m\in [n]$ such that 
\begin{enumerate}[(i)]
\item\label{it:firstpropertyexistence}
$m$ is $L$-good,
\item\label{it:secondpropertyexistence}
$m$ has the limited signaling property, and
\item\label{it:thirdpropertyexistence}
$m>n/2$ (i.e., $m$ belongs to the ``second half'' of input indices).
\end{enumerate}
To establish this statement, 
we first show a lower bound on the number of  indices~$m\in [n]$ satisfying~\eqref{it:firstpropertyexistence} and~\eqref{it:thirdpropertyexistence}. 
the following notation will be convenient: consider a subset~$\cG\subseteq [n]$ of input block indices. We will often partition such a set as
\begin{align}
\cG&=\cG_{\leq}\cup \cG_{>}\qquad\textrm{ where }\qquad 
\begin{matrix}
\cG_{\leq }&= &\{m\in \cG\ |\ m\leq n/2\}\\
\cG_{> }&= &\mspace{10mu}\{m\in \cG\ |\ m> n/2\}\ .
\end{matrix}
\end{align}
\begin{lemma}\label{lem:firstlemlineconeprop}
Set 
\begin{align}
    L&:=\left\lfloor\frac{n}{4 \ell}\right\rfloor\ .
\end{align}
Then 
\begin{align}
    |\mathsf{Good}(L)_{>}|\geq \frac{n}{4}\ .
\end{align}
\end{lemma}

\begin{proof}
By the definition of the set~$\mathsf{Bad}(L)_{>}$, we have 
\begin{align}
\cL^\rightarrow(m) \cap [L]&\neq \emptyset\qquad\textrm{ for every }\qquad 
m\in \mathsf{Bad}(L)_{>}\ .
\label{eq:badconditionim}
\end{align}
By the definitions of the forward and backward lightcones, condition~\eqref{eq:badconditionim} is equivalent to
\begin{align}
m\in \cL^\leftarrow\left([L]\right)\qquad\textrm{ for every }\qquad m\in \mathsf{Bad}(L)_{>}\ ,
\end{align}
that is,
\begin{align}
    \mathsf{Bad}(L)_{>}\subseteq \cL^\leftarrow\left([L]\right)\ .
\end{align}
It follows that 
\begin{align}
|\mathsf{Bad}(L)_{>}|&\leq \left|\cL^\leftarrow\left([L]\right)\right| \ .\label{eq:lowerboundloj}
\end{align}
By the locality-$\ell$ assumption on~$\mathsf{C}$, we have
\begin{align}
|\cL^\leftarrow(j)| &\leq \ell\qquad\textrm{ for all }\qquad j\in [n]\ 
\end{align}
and thus
\begin{align}
\left|\cL^\leftarrow\left([L]\right)\right|\leq L\ell\ .\label{eq:upperboundlmf}
\end{align}
Combining~\eqref{eq:upperboundlmf} with~\eqref{eq:lowerboundloj} yields
\begin{align}
|\mathsf{Bad}(L)_{>}| & \leq L\ell\ .
\end{align} or equivalently
\begin{align}
    |\mathsf{Good}(L)_{>}|\geq \lceil n/2\rceil - L\ell \geq n/4\ 
\end{align}
by the definition of~$L$.
\end{proof}

We next derive a lower bound on the number of  indices~$m\in [n]$ satisfying~\eqref{it:secondpropertyexistence} and~\eqref{it:thirdpropertyexistence}. 
The statement and the proof of the following result are analogous to those of~\cite[Claim~5]{bravyigossetkoenig}.  
\begin{lemma}\label{lem:secondlemlightconesize}
We have
\begin{align}
\big|\mathsf{LimitedSignaling}_{>}\big|&\geq 
\frac{3n}{8}\ .
\end{align}
\end{lemma}
\begin{proof}
Consider a bipartite graph $G=(V_I\cup V_O,E)$ with two sets of vertices
\begin{align}
    V_I&=\{v_i\ |\ i\in [n]_{>}\}\\
    V_O&=\{w_o\ |\ o\in[n]\}\ ,
\end{align}
where $(v_i,w_o)\in V_I\times V_O$ are connected by an edge if and only if
\begin{align}
    i\in \cL^{\leftarrow}(o)\ .\label{eq:edgeconditionone}
\end{align}
Observe that~\eqref{eq:edgeconditionone} is equivalent to
\begin{align}
    o\in \cL^{\rightarrow }(i)\ .\label{eq:xedgeconditionone}
\end{align}
Using~\eqref{eq:edgeconditionone}, the set of edges of this graph can be upper bounded by\footnote{Here and below we write~$\prod_{j=0}^k x_j=x_k\cdots x_0$ for non-commutative variables $x_0,\ldots,x_k$.}
\begin{align}
    |E|&\leq \sum_{o=1}^n |\cL^\leftarrow(o)|\leq  n\ell\ .\label{eq:upperboundejlo}
\end{align}
On the other hand, using the definition of the forward lightcone and~\eqref{eq:xedgeconditionone}, 
we have that for any $m\in \{\lfloor n/2+1\rfloor,\ldots,n\}\setminus \mathsf{LimitedSignaling}$, 
the subset $\{w_o\ |\ o\in \cL^\rightarrow(m)\}\subseteq V_O$ has more than~$8\ell$ neighbors in~$V_I$. This gives the lower bound
\begin{align}
    |E| &\geq \left|\{\lfloor n/2+1\rfloor,\ldots,n\}\setminus \mathsf{LimitedSignaling}_{>}\right|\cdot 8\ell\\
    &=\left(\lceil n/2\rceil - |\mathsf{LimitedSignaling}_{>}|\right)\cdot 8\ell\ .\label{eq:secondelowerbndx}
\end{align}
Combining~\eqref{eq:secondelowerbndx}  with~\eqref{eq:upperboundejlo} gives the claim.
\end{proof}

Finally, we can combine both properties to show that there is a large number of indices~$m\in [n]$ 
with $m>n/2$ (property~\eqref{it:thirdpropertyexistence}) that are $L$-good (cf.~\eqref{it:firstpropertyexistence}) and have the limited-signaling property (cf.~\eqref{it:secondpropertyexistence}). 
\begin{theorem}\label{thm:maincombinatorial}
Let $\mathsf{C}$ be a circuit with locality upper bounded by~$\ell$.
 Set 
\begin{align}
    L:=\left\lfloor\frac{n}{4\ell}\right\rfloor\ .
\end{align}
Then for sufficiently large $n$ there is some input index $m\in \{\lfloor n/2+1\rfloor,\ldots,n\}$
that is both $L$-good and has limited signaling, i.e.,
\begin{align}
     \cL^\rightarrow(m)\cap [L]&=\emptyset \label{eq:imlightconeoj}
\end{align}
and
\begin{align}
    |\cL^\rightarrow(m)| &\leq 8\ell\ .\label{eq:lightconesizeestimate}
\end{align}
\end{theorem}

\begin{proof}
This follows immediately by combining Lemma~\ref{lem:firstlemlineconeprop}
with Lemma~\ref{lem:secondlemlightconesize}. Indeed, by Lemma~\ref{lem:firstlemlineconeprop} we have 
\begin{align}
    |\mathsf{Good}(L)_{>}|&\geq n/4\ ,
\end{align}
and by Lemma~\ref{lem:secondlemlightconesize} we have
\begin{align}
   |\mathsf{LimitedSignaling}_{>}|&\geq 3n/8\ .
\end{align}
Since both $\mathsf{Good}(L)_{>}$ and $\mathsf{LimitedSignaling}_{>}$ are subsets of~$[n]_{>}$, we also have
\begin{align}
    |\mathsf{Good}(L)_{>}\cup \mathsf{LimitedSignaling}_{>}|\le \left|\,[n]_{>}\right|= \lceil n/2\rceil\ .
\end{align}
Consequently, we have
\begin{align}
   &|\mathsf{Good}(L)_{>}\cap \mathsf{LimitedSignaling}_{>}|\\ &\qquad\qquad = |\mathsf{Good}(L)_{>}| +|\mathsf{LimitedSignaling}_{>}| - |\mathsf{Good}(L)_{>}\cup \mathsf{LimitedSignaling}_{>}|\\ &\qquad\qquad\ge n/4+3n/8-\lceil n/2 \rceil\ge  n/8 -1\ .
\end{align}
This implies the claim. 
\end{proof}

\section{Quantum advantage against $\NC^0$}\label{sec:advantageagainstNCzero}
Our goal is to show that~$\psim{\teleport}\not\in \NC^0/\poly$, i.e., constant-depth quantum circuits have an advantage over (even non-uniform)~$\NC^0$ for the possibilistic simulation of the gate-teleportation circuit. 
Our proof is by contradiction: we show the implication 
\begin{align}
\left(
\psim{\teleport}\in \NC^0/\poly
\right)\qquad \Rightarrow \qquad 
\left(\parityproblem\in \NC^0/\poly\right)\ .\label{eq:implicationtoshowm}
\end{align}
This implication shows that~$\psim{\teleport}\not\in \NC^0/\poly$ since we would otherwise get a contradiction to the trivial result~$\parityproblem\not\in \NC^0/\poly$.

In Section~\ref{sec:simteleptosimcliffpq} we prove that any $\NC^0$ algorithm for the $\psim{\teleport_n}$ problem can be used to construct a non-uniform $\NC^0$ algorithm for the $\psim{\PCliff{Q}_L}$ problem for $L=\Theta(n)$ and  a fixed function~$Q$. In Section~\ref{sec:possibilistictomographyredc} we follow and specialize the techniques of~\cite{grier2020interactive}  to obtain a randomized $\NC^0$ algorithm that uses an oracle for $\psim{\PCliff{Q}_L}$ to learn some information about the stabilizer group of states  of the form~$(I\otimes C_L\cdots C_1)\ket{\Phi}$, where $\ket{\Phi}$ is a Bell state.  Finally, in Section~\ref{sec:quantumadvantage} we use this learning problem to solve $\parityproblem$. This establishes~\eqref{eq:implicationtoshowm}. 

Since the teleportation circuit $\teleport$ trivially solves the $\psim{\teleport}$ problem, 
this completes the proof of Theorem~\ref{thm:main}, i.e., establishes  a quantum advantage over~$\NC^0$.
\subsection{Relating $\psim{\teleport}$ to  $\psim{\PCliff{Q}}$}\label{sec:simteleptosimcliffpq}
Here we argue that an $\NC^0$-circuit $\mathsf{C}$
which possibilistically simulates the gate-teleportation circuit~$\teleport_n$ gives rise to a non-uniform~$\NC^0$-circuit (using oracle access to~$\mathsf{C}$) which simulates~$\PCliff{Q}_L$ 
for~$L=\Theta(n)$ and a suitably chosen function~$Q$. That is, we show the implication
\begin{align}
\left(\psim{\teleport}\in \NC^0/\poly\right) \qquad \Rightarrow \qquad \left(\exists Q\colon  \psim{\PCliff{Q}}\in \NC^0/\poly\right)\ .\label{eq:nczeroseparation2}
\end{align}

\subsubsection{Motivation\label{sec:motivation}}
To motivate the construction, recall what the effect of 
the (concurrent) Bell measurements 
in a gate-teleportation circuit is: it essentially amounts to the application of a product~$C_n\cdots C_1$ of Clifford group elements to one half of a maximally entangled state, but the resulting state~$(I\otimes C_n\cdots C_1)\ket{\Phi}$ is corrupted by a Pauli error~$\widetilde{Q}$ determined by the measurement outcomes. In standard uses of gate teleportation (such as in error correction), this Pauli error is computed and applied in a correction step. In contrast,
the circuit~$\teleport_n$ does not incorporate such a correction step.  We note that
the  computation of~$\widetilde{Q}$ from the measurement outcomes is not feasible using, e.g., an $\NC^0$ or an $\AC^0$-circuit.

Our construction avoids the necessity of computing Pauli corrections such as~$\widetilde{Q}$ by exploiting the fact that the correction function~$Q$ in the definition of the circuit~$\PCliff{Q}_L$  can be chosen arbitrarily (without computational complexity considerations). Very roughly, the circuit~$\teleport_n$ emulates 
the circuit~$\PCliff{\widetilde{Q}}_L$ for suitably chosen~$L$ and~$\widetilde{Q}$.

In more detail, let us  examine the relationship between the quantum circuits for~$\teleport_n$ and~$\PCliff{Q}_L$. 
The output distribution of~$\teleport_n$ is given by
\begin{align}
    p^{\teleport_n}(P_0,\ldots,P_{n-1}\mid C_0,\ldots, C_{n-1})&=\frac{1}{4^{n}}\left|\tr\left(P_{n-1}C_{n-1}\cdots P_1C_1P_0C_0\right)\right|^2 \ ,\label{eq:poutteleportn}
\end{align}
which we show at the end of this section, see Eq.~\eqref{eq:14 restated} and bellow.

Eq.~\eqref{eq:poutteleportn} can be written as\footnote{The derivation is similar to the one described for~\eqref{eq:poutteleportn}, see~\eqref{eq:14 restated} and bellow.}
 \begin{align}
 p^{\teleport_n}(P_0,\ldots,P_{n-1}\mid C_0,\ldots, C_{n-1})&=
 \frac{1}{4^{n}}\left|\tr\left(
 (P_{n-1}C_{n-1})
 (P_{n-2}C_{n-2}\cdots P_1C_1P_0C_0)
 \right)\right|^2\notag\\
 &=\frac{1}{4^{n}}\left|\tr\left(
 (P_{n-1}C_{n-1})\widetilde{Q} C_{n-2}\cdots C_0
 \right)\right|^2\ ,\label{eq:commutedversionoutputtelepn}
 \end{align}
 where the Pauli operator~
 \begin{align}
 \widetilde{Q}=\widetilde{Q}(C_0,\ldots,C_{n-2},P_0,P_1,\ldots,P_{n-2})\label{eq:tildeQdefinitionx}
 \end{align}
  is obtained by commuting all the Pauli operators in the product~$P_{n-2}C_{n-2}\cdots P_0C_0$ to the left, i.e.,
 \begin{align}
 \widetilde{Q}\prod_{i=0}^{n-2}C_i &=P_{n-2}C_{n-2}\cdots P_0C_0\ .
 \end{align}
On the other hand, the output distribution of~$\PCliff{Q}_L$ is 
\begin{align}
    p^{\PCliff{Q}_L}(P\ |\ C_1,\ldots,C_L,D)&=
\left|\bra{\Phi}(P\otimes I)(I\otimes \left(DQ(C_1,\ldots,C_L)(C_L\cdots C_1)\right))\ket{\Phi}\right|^2\ \label{eq:pcliffqoutputdistributionm} 
\end{align}
which can be written as 
\begin{align}
 p^{\PCliff{Q}_L}(P\mid C_1,\ldots, C_{L},D)&=\frac{1}{4}\left|\tr\left(PDQC_L\cdots C_1\right)\right|^2\ ,\ \textrm{where}\  Q=Q(C_1,\ldots,C_L)\ .\label{eq:poutPCliffQL}
 \end{align}
  Note that we use $Q=Q(C_1,\ldots,C_L)$ to indicate that $Q$ is a function of $(C_1,\ldots,C_L)$ and of no other arguments.
  Comparing~\eqref{eq:commutedversionoutputtelepn}
 and~\eqref{eq:poutPCliffQL}, setting $n:=L+1$ and 
 \begin{align}
     C_j'&:=\begin{cases}
     C_{j+1}\qquad&\textrm{ for }\quad j=0,\ldots,n-2\\
     D &\textrm{ for }\quad j=n-1\ 
     \end{cases}\label{eq:specialchoiselcnp}
 \end{align} we have 
  the implication
 \begin{align}
 p^{\teleport_n}(P_0,\ldots,P_{n-1}\mid 
 C_0',\ldots,C_{n-1}')>0\quad \Rightarrow\quad  \tilde{p}(P\mid C_1,\ldots,C_L,D)>0\ ,\label{eq:implicationm}
 \end{align}
 where we have written $P$ for $P_{n-1}$, and 
where the distribution~$\tilde{p}$
is defined as in~\eqref{eq:poutPCliffQL}
but with the Pauli operator~$Q(C_1,\ldots,C_L)$ replaced by
$\widetilde{Q}(C_1,\ldots,C_{L},P_0,\ldots,P_{L-1})$, that is,
\begin{align}
\tilde{p}(P\mid C_1,\ldots, C_{L},D)&:=
\frac{1}{4}\left|\tr\left(PD\widetilde{Q}C_L\cdots C_1\right)\right|^2\, ,
\label{eq:tildepm}
\end{align}
where we used a shorthand notation $\widetilde{Q}=\widetilde{Q}(C_1,\ldots,C_L,P_0,\ldots,P_{L-1})$.
Eq.~\eqref{eq:implicationm} and comparing~\eqref{eq:tildepm} with~\eqref{eq:poutPCliffQL} suggests that 
a possibilistic simulation of~$\teleport_n$ can be used to possibilistically simulate~$\PCliff{Q}_L$ for a suitably chosen function~$Q$. However, 
the fact that $\widetilde{Q}$ also depends on~$P_0,\ldots,P_{L-1}$ needs to be addressed---the expression~\eqref{eq:tildepm} differs from the output distribution~\eqref{eq:poutPCliffQL} of~$\PCliff{Q}_L$.  
This is where no-signaling considerations enter the argument.

It remains to show the technical claim of Eq.~\eqref{eq:poutteleportn} that we restate and prove here.
\begin{align}
	 p^{\teleport_n}(P_0,\ldots,P_{n-1}\mid C_0,\ldots, C_{n-1})&=\frac{1}{4^{n}}\left|\tr\left(P_{n-1}C_{n-1}\cdots P_1C_1P_0C_0\right)\right|^2 \label{eq:14 restated}
\end{align}
	The correctness can be easily verified using  the tensor-network formalism. Alternatively, the identity can be obtained as follows. Let $\pi_N:(\mathbb{C}^2)^{\otimes N}\rightarrow (\mathbb{C}^2)^{\otimes N}$ be the unitary which cyclically ``left'' permutes tensor factors, i.e., $\pi_N$ is defined as $\pi_N(\ket{\Psi_0}\otimes\cdots\otimes \ket{\Psi_{N-1}})=\ket{\Psi_{1}}\otimes \ket{\Psi_2}\otimes \cdots \otimes \ket{\Psi_{N-1}}\otimes \ket{\Psi_{0}}$ for  $\ket{\Psi_j}\in\mathbb{C}^2$, $j\in \mathbb{Z}_N$, and linearly extended to all of $(\mathbb{C}^2)^{\otimes N}$.  This satisfies
	\begin{align}
		\tr(\pi_N (A_0\otimes\cdots\otimes  A_{N-1}))=\tr(A_0\cdots A_{N-1})\ \label{eq:cyclidefinitionf}
	\end{align}
	for any $N$-tuple of operators~$(A_0,\ldots,A_{N-1})$ on $\mathbb{C}^2$.
	Also observe that
	\begin{align}
		\cE(A):=\frac{1}{4}\sum_{k,\ell\in \{0,1\}} X^kZ^\ell AX^kZ^\ell&=\frac{1}{2} A^T\qquad\textrm{ for any operator $A$ on $\mathbb{C}^2$\ .}\label{eq:eaveragem}
	\end{align}
	We  use that the projection onto the maximally entangled state~$\ket{\Phi}$ can be written as
	\begin{align}
		\proj{\Phi}&=\frac{1}{4}(I\otimes I+X\otimes X-Y\otimes Y+Z\otimes Z)=
		\frac{1}{4}\sum_{a,b\in\{0,1\}} (X^aZ^b)\otimes (X^aZ^b)\ .\label{eq:phystateexpression}
	\end{align}
	The probability~$p^{\teleport_n}(P_0,\ldots,P_{n-1}\mid C_0,\ldots, C_{n-1})$ of interest 
	is the squared overlap of the state
	$\left((I\otimes  C_0)\otimes \dots \otimes (I\otimes  C_{n-1})\right)\ket{\Phi}^{\otimes n}$ 
	and the state
	\begin{align}
	\pi^\dagger_{2n}\left((P_0\otimes  I)\otimes \dots \otimes (P_{n-1}\otimes  I)\right)\ket{\Phi}^{\otimes n} \ ,
	\end{align}
	 see Fig.~\ref{fig:teleportationcircuit}. The latter is equivalent to 
	\begin{align}
	 &\pi^\dagger_{2n}\left((P_0\otimes  I)\otimes \dots \otimes (P_{n-1}\otimes  I)\right)\pi_{2n}\pi^\dagger_{2n}\ket{\Phi}^{\otimes n}\\
	 &=\left((I\otimes  P_0)\otimes \dots \otimes (I\otimes  P_{n-1})\right)\pi^\dagger_{2n}\ket{\Phi}^{\otimes n} \ .
	 \end{align}
	Thus the overlap of interest can be expressed as
	\begin{align}
			&\left|\tr\left(
			\left( (I\otimes  P_0C_0)\otimes \dots \otimes (I\otimes  P_{n-1}C_{n-1}) \right)
			\proj{\Phi}^{\otimes n}\pi_{2n}\right)\right|^2 \\
			&=\Big|\frac{1}{4^n}
			\sum_{\vec{a},\vec{b} \in \{0,1\}^n} \tr \Big(\pi_{2n} \big(( X^{a_0}Z^{b_0} \otimes P_0C_0 X^{a_0}Z^{b_0} ) \otimes \dots \notag\\[-0.2cm]
&\qquad\qquad\qquad\qquad\qquad			\dots\otimes (X^{a_{n-1}}Z^{b_{n-1}} \otimes P_{n-1}C_{n-1} X^{a_{n-1}}Z^{b_{n-1}}))\big)\Big)
			\Big|^2 \quad \textrm{by~\eqref{eq:phystateexpression}}\\
			&=\left|\tr\left(\cE(P_0C_0)\cdots \cE(P_{n-1}C_{n-1})\right)\right|^2 \qquad \textrm{\hfil by~\eqref{eq:cyclidefinitionf}}\\
			&=\frac{1}{2^{2n}}\left|
			\tr\left((P_0C_0)^T\cdots (P_{n-1}C_{n-1})^T\right)\right|^2
			\qquad \textrm{ by \eqref{eq:eaveragem}}\\
			&=\frac{1}{2^{2n}}\left|\tr\left(
			((P_{n-1}C_{n-1})\cdots (P_0C_0))^T
			\right)\right|^2\ .
	\end{align}
	The claim~\eqref{eq:14 restated} (cf. \eqref{eq:poutteleportn}) now follows because the transpose is trace-preserving, i.e., $\tr(A^T)=\tr(A)$.

\subsubsection{An algorithm for~$\psim{\PCliff{Q}}$\label{sec:algorithmpsimcliffql}}
Suppose $\mathsf{C}\colon\cC^n\rightarrow\cP^n$ is a deterministic oracle (i.e., a function)
that solves the problem $\psim{\teleport_n}$.  In this section, we argue that query access to~$\mathsf{C}$ 
allows  an $\NC^0$-circuit to solve the problem~$\psim{\PCliff{Q}_L}$ for some fixed function~$Q$, assuming that~$\mathsf{C}$ has certain locality properties we detail below. Moreover, this property  matches exactly what we established in Section~\ref{sec:nosignalingnc0} for~$\NC^0$-circuits.
The main idea behind the reduction is identical to that explained in the previous section: we embed instances of~$\psim{\PCliff{Q}_L}$ into~$\psim{\teleport_n}$. However, we need to deal with {\em some} amount of signaling
since the $\widetilde{Q}$ in Eq.~\eqref{eq:tildepm} could be indirectly influenced by $D$ through the Pauli operators $(P_0,\ldots, P_{L-1})$.
Therefore, contrary to the choice of embedding~\eqref{eq:specialchoiselcnp}, our embeddings use longer sequences (i.e., large values of~$n$), essentially ``padding'' with identity operators.
 
To define these embeddings, suppose that there are $L$ and $m$ such that $L\leq m <n$, which  satisfy properties illustrated in Figure~\ref{fig:nosignalingstructure}.
     \begin{figure}[!h]
         \centering
         \includegraphics{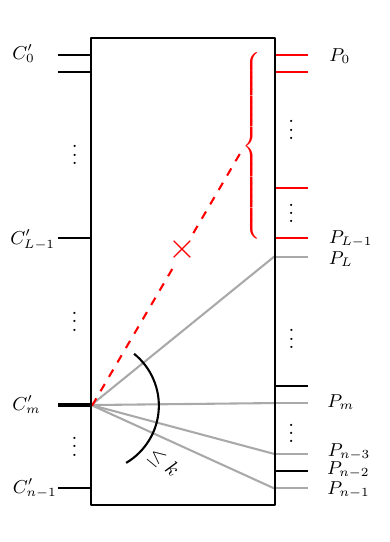}
         \caption{An illustration of the no-signaling assumption on the oracle for the $\psim{\teleport_n}$ problem. We assume the input $C'_m:=D$ cannot signal to outputs $P_0,\ldots,P_{L-1}$ and can only signal to a constant $k$ number of outputs with indices from the set $\cJ$. We depict the outputs influenced by the input $C'_m$ in grey.}
         \label{fig:nosignalingstructure}
     \end{figure}
     
We  embed an instance~$(C_1,\ldots,C_L,D)\in \cC^{L+1}$  of the problem $\psim{\PCliff{Q}_L}$ (where $Q:\cC^L\rightarrow\cP$ can be any fixed function)
into an instance~$(C'_0,\ldots,C'_{n-1})\in\cC^n$ of~$\psim{\teleport_n}$ as follows: we set
\begin{align}
 C'_j &=\begin{cases}
 C_{j+1}\qquad &\textrm{ for }\quad j=0,\ldots,L-1\ ,\\
 I &\textrm{ for }\quad j\geq L,\ j\neq m\ ,\\
 D &\textrm{ for }\quad j=m\ .
 \end{cases}\label{eq:mainembeddingideasimcliffql}
\end{align}
To examine the result of this embedding, we need to consider the output probabilities 
for an embedded instance: 
we have 
\begin{align}
\!\!\!\!\!\!\!\!\!\!\! p^{\teleport_n}(P_0,\makebox[1em][c]{.\hfil.\hfil.}\,,P_{n-1}\!\mid\! 
 C_0',\makebox[1em][c]{.\hfil.\hfil.}\,,C_{n-1}')\!&=\!\frac{1}{4^{n}}\left|\tr\!\left(\!\left(\prod_{r=m}^{n-1} P_r\!\right)
		\!D\!
		\left(\prod_{s=L}^{m-1}P_s\!\right)\!\!\left(\prod_{t=1}^L
		(P_{t-1} C_t)\!\right)\!\right)\right|^2\!\! . \label{eq:explicitexpressionpteleportn}
\end{align}
We display the trace defining the probability $p^{\teleport_n}(P_0,\ldots,P_{n-1}|C'_0,\ldots,C'_{n-1})$ as well as the same trace after we inserted variables from the embedding (i.e. the right-hand side of Eq.~\eqref{eq:explicitexpressionpteleportn}) in tensor-network diagrammatic representation in Figure~\ref{fig:gif1gif2}.

\begin{figure}[!h]
     \centering
         \begin{subfigure}[b]{\textwidth}
         \centering
         \includegraphics[width=\textwidth]{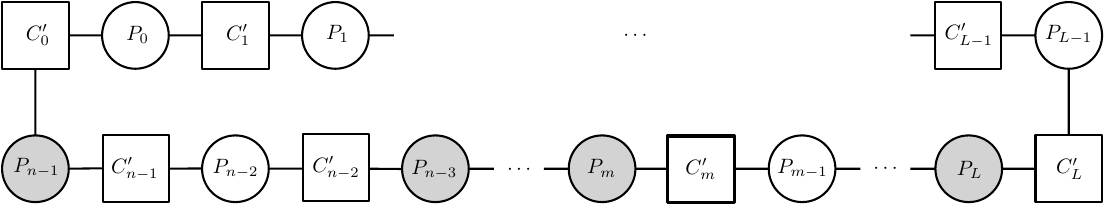}
         \caption{Tensor-network representation of the trace defining $p^{\teleport_n}(P_0,\ldots,P_{n-1}| C'_0,\ldots, C'_{n-1})$, from Eq.~\eqref{eq:poutteleportn}. We depict Clifford unitaries by squares and Pauli operators by circles. The output Pauli operators influenced by the input $C'_m$ are illustrated in grey.}
         \label{fig:nosignalingtrace}
         \end{subfigure}
         \\[0.7cm]
         \begin{subfigure}[b]{\textwidth}
         \centering
        \includegraphics[width=\textwidth]{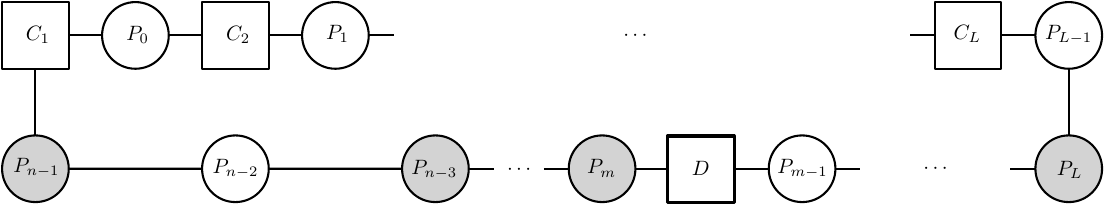}
        \caption{The trace defining the output probability distribution of $\teleport_n$ after we insert the inputs from Eq.~\eqref{eq:mainembeddingideasimcliffql} (also later in Algorithm~\ref{alg:simcliffptosimtelepalg} after the lines~\eqref{step:redAinit}--\eqref{step:redAendinit}). }
        \label{fig:gifplugingvalues}
         \end{subfigure}
        \caption{Tensor-network diagrammatic representation of (a) the trace that defines $p^{\teleport_n}(P_0,\ldots,P_{n-1}| C'_0,\ldots, C'_{n-1})$, from Eq.~\eqref{eq:poutteleportn}, and (b) the same trace after we inserted variables from the embedding~\eqref{eq:mainembeddingideasimcliffql}. Note that this figure illustrates a case where $m>L$, however, our arguments also apply for $m=L$ without further modifications.}
        \label{fig:gif1gif2}
\end{figure}

We need to relate this to the output probability~$p^{\PCliff{Q}_L}(P\mid C_1,\ldots, C_{L},D)$ of the circuit~$\PCliff{Q}_L$, see Eq.~\eqref{eq:poutPCliffQL}.
 It is easy to see (using commutation relations) that expression~\eqref{eq:explicitexpressionpteleportn} can indeed be related to~$p^{\PCliff{Q}_L}(P\mid C_1,\ldots, C_{L},D)$
 for suitable choices of~$Q$ and~$P$. The less obvious a priori is that, when going from~$(P_0,\ldots,P_{n-1})$ to~$(Q,P)$, the dependence on the Paulis~$(P_{j_1},\ldots,P_{j_k})$  for a fixed subset~$\cJ:=\{j_1<\cdots <j_k\}$ of positions can be limited to~$P$ (whereas~$Q$ can be chosen not to depend on these Paulis). This is a consequence of the cyclic structure of the circuit~$\teleport_n$. We formalize it in the following Lemma~\ref{lem:technicalcommutinglemma}.

 \begin{lemma}\label{lem:technicalcommutinglemma}
 Let $L\leq m<n$ be integers. Let 
 \begin{align}
     \cJ=\{j_1<\cdots <j_k\}
 \end{align}
 be a subset of $\{L,L+1,\ldots,n-1\}$. Then there are two functions
 \begin{align}
 \begin{matrix}
     P:&\cP^k\times\cC & \rightarrow &\cP\\
     Q:&\cP^{n-k}\times\cC^L & \rightarrow & \cP
     \end{matrix}
 \end{align}
 such that 
 	\begin{align}
		\left|\tr\left(\left(\prod_{r=m}^{n-1} P_r\right)
		D
		\left(\prod_{s=L}^{m-1}P_s\right)\left(\prod_{t=1}^L
		(P_{t-1} C_t)\right)\right)\right|=\left|\tr\left(PDQ\prod_{j=1}^L C_j \right)\right| \  \label{eq:commutinglemma}
	\end{align}
	for all $(P_0,\ldots,P_{n-1})\in\cP^n$ and $(C_1,\ldots,C_{L},D)\in\cC^{L+1}$, where 
	\begin{align}
	    P&=P\left(\left(P_j\right)_{j\in\cJ},D\right)\label{eq:pdefinitiondependence}\\
	     Q&=Q\left(\left(P_j\right)_{j\not\in\cJ},C_1,\ldots,C_L\right)\ .\label{eq:qdefinitiondependence}
	\end{align}
	Here we write~$\left(P_j\right)_{j\in\cJ}$ for $(P_{j_1},\ldots,P_{j_k})$
	(and similarly for $\left(P_j\right)_{j\not\in\cJ}$).
 \end{lemma}
 We illustrate each step in the proof of Lemma~\ref{lem:technicalcommutinglemma}  in the tensor-network diagrammatic representation for clarity. The left-hand side of Eq.~\eqref{eq:commutinglemma} is shown  in Figure~\ref{fig:gifplugingvalues}.
 \begin{proof}
 Let $(P_0,\ldots,P_{n-1})\in\cP^n$ and $(C_1,\ldots,C_{L},D)\in\cC^{L+1}$ be arbitrary. Define
	\begin{align}
		R&:=\left(\prod_{r=m}^{n-1} P_r\right) \qquad\text{ and }\qquad S:=\left(\prod_{s=L}^{m-1}P_s\right)\ .
	\end{align}
	Then the quantity of interest is equal to
	\begin{align}
		\left|\tr\left(\left(\prod_{r=m}^{n-1} P_r\right)
		D
		\left(\prod_{s=L}^{m-1}P_s\right)\left(\prod_{t=1}^L
		(P_{t-1} C_t)\right)\right)\right|&=
		\left|\tr\left(
		R
		D
		S
		\left(\prod_{t=1}^L
		(P_{t-1} C_t)\right)\right)\right|\label{eq:commutingrdsptct}\ .
	\end{align}
	 Since Pauli operators either commute or anticommute, we can factor the products of Pauli operators~$R$ and~$S$ into factors $P_j$ with $j\in \cJ$, and factors~$P_j$ with $j\in\mathbb{Z}_n\backslash \cJ$. That is, we have	
		\begin{align}
		R=\pm Q'P' \qquad \text{ and }\qquad S=\pm P''Q'' \ ,\label{eq:rsfactorizationqprimepprime}
	\end{align}
	where
	\begin{align}
		\begin{matrix}
			P'&:=&\prod_{r=m}^{n-1}&
			P_r^{\delta_{r}}\\
			Q'&:=&\prod_{r=m}^{n-1}&P_r^{1-\delta_r}
		\end{matrix}\qquad\textrm{ and }\qquad 
		\begin{matrix}
			P''&:=&\prod_{s=L}^{m-1}&P_s^{\delta_s}\\
			Q''&:=&\prod_{s=L}^{m-1}&P_s^{1-\delta_s}\ 
		\end{matrix}\label{eq:clppqq}
	\end{align}
	where  $\delta_r\in \{0,1\}$ indicates whether or not~$r\in\cJ$, that is, 
	\begin{align}
		\delta_r&:=\begin{cases}
			1 \qquad &\textrm{ if }r\in\cJ\\
			0&\textrm{ otherwise}\ .
		\end{cases}
	\end{align}
	Inserting the factorization~\eqref{eq:rsfactorizationqprimepprime} into expression~\eqref{eq:commutingrdsptct} we have
	\begin{align}
		\pm\tr\left(RDS\left(\prod_{t=1}^L
		(P_{t-1} C_t)\right)\right)&=
    \tr\left((Q'P')D(P''Q'') 
		\left(\prod_{k=1}^L
		(P_{k-1} C_k)\right)\right)\\
		&= \tr\left(Q'P'DP''Q'' Q'''\left(\prod_{j=1}^L C_j\right)\right)\ .
		\label{eq:clpdpqqcjq}
	\end{align}
	Here we commuted the Pauli operators $P_0,\ldots,P_{L-1}$  in the product~$\prod_{t=1}^L
		(P_{t-1} C_t)$ to the left, combining them into a Pauli operator~$Q'''\in\cP$ defined by 
	\begin{align}
Q''':=\left(\prod_{r=1}^L (P_{r-1}C_r)\right)\left(\prod_{r=1}^L C_r\right)^{-1}\ .
	\end{align}
	At this point, we use the cyclicity of the trace to move~$Q'$ to the right of the product in the argument, obtaining
	\begin{align}
	 \tr\left(Q'P'DP''Q'' Q'''\left(\prod_{j=1}^L C_j\right)\right)&=\tr\left(P'DP''Q'' Q'''\left(\prod_{j=1}^L C_j\right)Q'\right)\\
	 &=\tr\left(P'DP''Q\left(\prod_{j=1}^L C_j\right)\right)\ ,\label{eq:pprimdpddprime}
	\end{align}
	where we defined 
	\begin{align}
		Q&:=Q''Q'''\left(\prod_{j=1}^L C_j\right) Q'\left(\prod_{j=1}^L C_j\right)^{-1}\ . \label{eq:collectingQ} 
	\end{align}
	Using the cyclicity of the trace to move~$Q'$ to the right (instead of commuting it past~$P'DP''$) here ensures that the resulting operator~$Q$ has no dependence on~$D$. Indeed, it is clear from the definition of $Q',Q'',Q'''$ and~\eqref{eq:collectingQ} that~$Q$ is a function of $(P_j)_{j\in\cJ}$ and $(C_1,\ldots,C_L)$ only as claimed, see Eq.~\eqref{eq:qdefinitiondependence}. We illustrate these  steps (commuting Paulis) in Figure~\ref{fig:gif3}.
      \begin{figure}[!h]
         \centering
         \includegraphics[width=\textwidth]{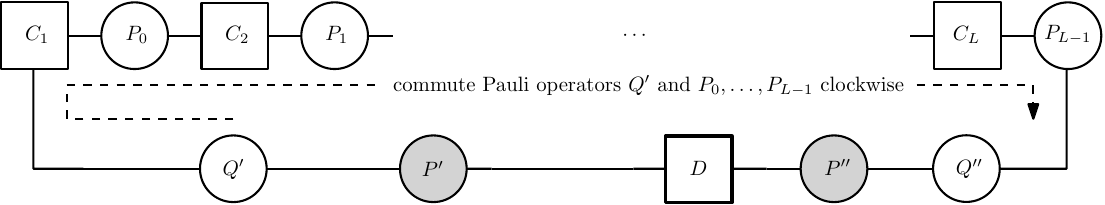}
         \caption{The use of the cyclic property of trace to commute Pauli operators not influenced by the input $D$ in a way that does not create additional dependence of the resulting Pauli operator on $D$. This figure illustrates the steps from Eq.~\eqref{eq:commutingrdsptct} to Eq.~\eqref{eq:pprimdpddprime}, after the Pauli operators composing $Q',P',P'',Q''$ have been respectively regrouped.}
         \label{fig:gif3}
     \end{figure}
	
	We can rewrite~\eqref{eq:pprimdpddprime}
as	\begin{align}
	\tr\left(P'DP''Q\left(\prod_{j=1}^L C_j\right)\right)&=
        \tr\left(PDQ\left(\prod_{j=1}^L C_j\right)\right) \label{eq:clpdpqcj}
	\end{align}
	by defining
    \begin{align}
        P:=P'DP''D^{-1}\ .\label{eq:clP}
    \end{align}
    It is easy to check using the definition of $P'$ and $P''$ and~\eqref{eq:clP} that $P$ is a function of only $(P_j)_{j\in\cJ}$ and~$D$ as claimed (see Eq.~\eqref{eq:pdefinitiondependence}). We illustrate the computation of $P$ in Figure~\ref{fig:gif4}.
    \begin{figure}[!ht]
         \centering
         \includegraphics[width=\textwidth]{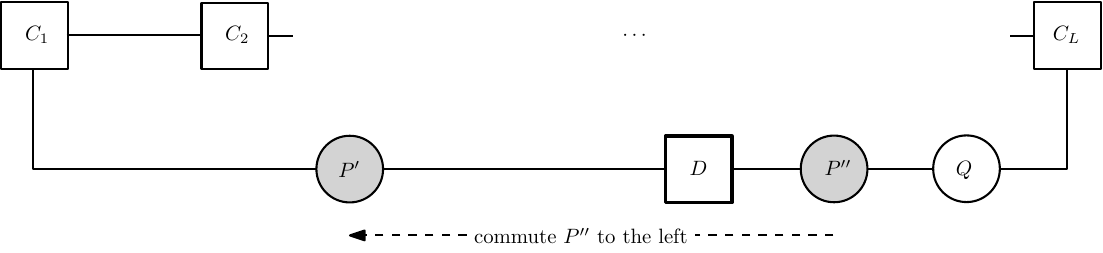}
         \caption{In this step we commute all Pauli operators depending on $D$ to the left of it. This illustrates the Eq.~\eqref{eq:clpdpqcj} and the step \eqref{step:computeP} of Algorithm~\ref{alg:Afunction}.}
         \label{fig:gif4}
     \end{figure}
    
    This proves the lemma since we have shown that~\eqref{eq:commutinglemma} is satisfied with these definitions, and we can assume without loss of generality that~$P,Q\in\cP$ (as~\eqref{eq:commutinglemma} includes absolute values). The resulting trace is depicted in Figure~\ref{fig:gif5}.
    \begin{figure}[!ht]
         \centering
         \includegraphics[width=\textwidth]{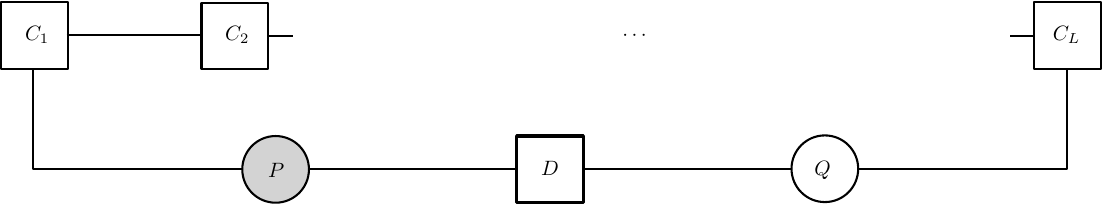}
         \caption{The resulting trace defining the probability $p^{\PCliff{Q}_L}(P\mid C_1,\ldots, C_{L},D)$ and therefore the solution to $\psim{\PCliff{Q}}$ problem. This illustrates the trace from Eq.~\eqref{eq:poutPCliffQL} to which we arrive at Eq.~\eqref{eq:clpdpqcj}.}
         \label{fig:gif5}
    \end{figure}
    \end{proof}

\begin{algorithm}
\caption{Algorithm defining the  function~$\cA_{\cJ,m}:\cP^n\times\cC\rightarrow\cP$. Here~$\cJ\subseteq \mathbb{Z}_n$ is a subset and~$m\in\mathbb{Z}_n$.
\label{alg:Afunction}}
\begin{algorithmic}[1]
       \Function{$\cA_{\cJ,m}$}{$P_0,\ldots, P_{n-1}, D$}
             \State $P'\leftarrow \prod_{p=m}^{n-1} P_p^{\delta_p}$ \quad where $\delta_p=1$ if  $p\in\cJ$ and $\delta_p=0$ otherwise
             \State $P''\leftarrow \prod_{q=L}^{m-1} P_q^{\delta_q}$ \label{step:Pdoubleprime}
            \State $P\leftarrow 
            P'DP''D^{-1}$. \label{step:computeP}
            \State Ignore (remove) the global phase to ensure that~$P\in\cP$.
            \State \Return $P$
        \EndFunction
\end{algorithmic}
\end{algorithm}

In the following, we will use that~$P$ only depends on the subsets of Paulis~$(P_j)_{j\in\cJ}$ with indices belonging to~$\cJ$ and the Clifford~$D$ as expressed by Eq.~\eqref{eq:pdefinitiondependence}. In fact, the Pauli~$P$ can easily be computed: the proof of Lemma~\ref{lem:technicalcommutinglemma} also shows that
the algorithm~$\cA_{\cJ,m}$ (Algorithm~\ref{alg:Afunction}) computes this function, i.e.,
we have 
\begin{align}
    \cA_{\cJ,m}(P_0,\ldots,P_{n-1},D)&=P\left(
    \left(P_j\right)_{j\in\cJ},D\right)\ \textrm{ for all }\ (P_0,\ldots,P_{n-1})\in\cP^n\textrm{ and }D\in\cC\,.\label{eq:equalitycomputationcbjmnP}
\end{align}
Importantly, the algorithm~$\cA_{\cJ,m}$ can be implemented efficiently. 

\begin{lemma}\label{lem:technicalcommutingefficientcomputation}
Let $k$ be a constant, let $L=\mathsf{poly}(n)<n$ and $m$ be such that $L\leq m<n$. Let $\cJ\subseteq \mathbb{Z}_n$ be of size $|\cJ|=k$. Then Algorithm~\ref{alg:Afunction} for computing the function~$\cA_{\cJ,m}:\cP^n\times\cC\rightarrow\cP$  can be implemented by an $\NC^0$-circuit. 
\end{lemma}
\begin{proof} 
  This follows immediately from equations~\eqref{eq:clppqq} and~\eqref{eq:clP}. Since $k$ is constant, the corresponding products can be computed by an $\NC^0$-circuit, see Lemma~\ref{lem:technicalcommutinglemma}. 
\end{proof}

\begin{algorithm}
\caption{The algorithm~$\cB_{L,m,\cJ}^{\mathsf{C}}$, which takes as
input an $L+1$-tuple~$(C_1,\ldots,C_L,D)\in \cC^{L+1}$ of Clifford group elements. It uses an oracle~$\mathsf{C}$ for the $\psim{\teleport_n}$ problem. 
For a suitable choice of~$(L,m,\cJ)$, the output~$P$ of~$\cB_{L,m,\cJ}^{\mathsf{C}}$
is a valid solution to 
$\psim{\PCliff{Q}_L}$ for a certain function~$Q:\cC^{L}\rightarrow\cP$, see Lemma~\ref{lem:simpcliffQtosimtelepnonsignaling}.
\label{alg:simcliffptosimtelepalg}}
\begin{algorithmic}[1]
        \Function{$\cB_{L,m,\cJ}^{\mathsf{C}}$}{$C_1,\ldots,C_L,D$}
            \For{$j=0,\ldots,L-1$}\label{step:redAinit}
                \State $C_j'\gets C_{j+1}$
            \EndFor
            \For{$j=L,\ldots,n-1$}
                \If{$j\neq m$}
                   \State $C_j'\gets I$
                 \Else
                   \State $C_j'\gets D$
                \EndIf 
            \EndFor\label{step:redAendinit}
            \State $(P_0,\ldots, P_{n-1})\leftarrow \mathsf{C}(C_0',\ldots, C'_{n-1})$\label{it:reductioner}
            \State $P\gets 
            \cA_{\cJ,m}(P_0,\ldots, P_{n-1},D)$ \label{step:POut}
            \State \Return $P$
        \EndFunction
\end{algorithmic}
\end{algorithm}

Our main result of this section combines Lemma~\ref{lem:technicalcommutinglemma} and Lemma~\ref{lem:technicalcommutingefficientcomputation} to show that an $\NC^0$~circuit~$\mathsf{C}$
 for $\psim{\teleport_n}$ can be used to solve~$\psim{\PCliff{Q}_L}$. To arrive at this result,  we make a suitable choice of~$L,m$ and the subset~$\cJ\subseteq \mathbb{Z}_n$ depending on the oracle~$\mathsf{C}$ in a non-black-box manner. More precisely,  we use the no-signaling properties of~$\mathsf{C}$ established in Section~\ref{sec:nosignalingnc0}. The resulting reduction is therefore non-uniform due to three parameters $(L,m,\cJ)$ that must be given as advice.
 The following theorem formalizes this and shows that  a suitable choice always exists.

\begin{lemma} \label{lem:simpcliffQtosimtelepnonsignaling}
    Let $\mathsf{C}\colon\cC^n\rightarrow \cP^n$ be a (possibly non-uniform) $\NC^0$-circuit for the $\psim{\teleport_n}$ problem. 
        Then  there are $(L,m,\cJ)$ with  $L=\Theta(n)$, $L\leq m<n$ and
        $\cJ\subseteq \mathbb{Z}_n$, as well as a function~$Q:\cC^L\rightarrow \cP$ such that the algorithm $\cB_{L,m,\cJ}^\mathsf{C}$ (Algorithm~\ref{alg:simcliffptosimtelepalg}) 
        has the following properties: 
        \begin{enumerate}[(i)]
        \item\label{it:firstclaimcomputationaclm}
        Given an arbitrary input~$(C_1,\ldots,C_L,D)\in\cC^{L+1}$,
        the output produced by the algorithm $P=\cB_{L,m,\cJ}^\mathsf{C}(C_1,\ldots, C_L, D)$ is a valid solution to~$\psim{\PCliff{Q}_L}$ associated with  the instance~$(C_1,\ldots,C_L,D)$.
        \item\label{it:computationaclmj}
        The algorithm $\cB_{L,m,\cJ}^\mathsf{C}$ can be implemented by an $\NC^0$-circuit making one oracle query to~$\mathsf{C}$.
        \end{enumerate}
\end{lemma}
\begin{proof}
    Suppose that~$\mathsf{C}\colon\cC^n\rightarrow\cP^n$ is an $\NC^0$-circuit. Let us denote input- and output-variables of~$\mathsf{C}$ by~$(C'_0,\ldots,C'_{n-1})$ and~$(P_0,\ldots,P_{n-1})$, respectively, i.e., we write
    \begin{align}
        (P_0,\ldots,P_{n-1})&=\mathsf{C}(C'_0,\ldots,C'_{n-1})\ .
    \end{align}
    We note that by definition,  the algorithm~$\cB_{L,m,\cJ}^\mathsf{C}$ uses the embedding~\eqref{eq:mainembeddingideasimcliffql}, i.e., 
    for an instance~$(C_1,\ldots,C_L,D)$ of $\psim{\PCliff{Q}_L}$, it calls the oracle~$\mathsf{C}$ with~$(C'_0,\ldots,C'_{n-1})$ defined by~\eqref{eq:mainembeddingideasimcliffql}.

        As a constant-depth circuit with bounded fan-in gates,  the circuit~$\mathsf{C}$ 
    is $\ell$-local with locality bounded by $\ell=O(1)$, see Eq.~\ref{eq:nczerocircuitbounddepth}.     Let $L:=\lfloor n/(4\ell)\rfloor$. Then $L=\Theta(n)$ by definition, and, as shown in Theorem~\ref{thm:maincombinatorial} there is an integer~$m\in \{ \lfloor n/2+1\rfloor,\ldots,n-1\}$  and a subset~$\cJ\subseteq\{L,\ldots,n-1\}$ such that
    \begin{enumerate}[(a)]
        \item \label{it:apartproofmain}
        the collection of output variables~$\left(P_j\right)_{j=0,\ldots,L-1}$ does not depend on the input variable~$C'_m$, i.e., 
        $\cL^\rightarrow(m)\cap \{0,\ldots,L-1\}=\emptyset$,
                and
        \item\label{it:bpartproofmain} 
        the input variable~$C'_m$ influences  at most~$8\ell$~output variables~$\{P_j\}_{j\in\cJ}$, i.e., we have $|\cJ|\leq 8\ell=O(1)$ for $\cJ=\cL^\rightarrow(m)$.
                    \end{enumerate}
    Statement~\eqref{it:bpartproofmain}  implies that the function~$\cA_{\cJ,m}$
    defined by Algorithm~\eqref{alg:Afunction}
    can be computed by an $\NC^0$-circuit from $(P_0,\ldots,P_{n-1})$, see Lemma~\ref{lem:technicalcommutingefficientcomputation}.  Claim~\eqref{it:computationaclmj} follows from this.

It remains to prove claim~\eqref{it:firstclaimcomputationaclm}. 
Note that~$\cA_{\cJ,m}$ computes the function
\begin{align}
    P\left(\left(P_j\right)_{j\in\cJ},D\right)
    \end{align}appearing in Lemma~\ref{lem:technicalcommutinglemma}, see 
Eq.~\eqref{eq:equalitycomputationcbjmnP}. Furthermore, the Pauli
\begin{align}
    Q&=Q\left(\left(P_j\right)_{j\not\in\cJ},C_1,\ldots,C_L\right)\label{eq:Qpjjcjonel}
\end{align}
 in Lemma~\ref{lem:technicalcommutinglemma} 
 does not depend on~$C'_m$ by~\eqref{it:apartproofmain}. 
 
 Consider an output variable $P_j$ with $j\not\in\cJ$. By definition, $P_j$ is a function of $\left(C'_j\right)_{j\in \mathbb{Z}_n\backslash \{m\}}$. Because the embedding~\eqref{eq:mainembeddingideasimcliffql} fixes to the identity all but the first~$L$ Cliffords~$(C'_0,\ldots,C'_{L-1})=(C_1,\ldots,C_L)$ and the $m$-th Clifford~$C'_m=D$, we conclude that $P_j$ (for $j\not\in\cJ$) is a function of $(C_1,\ldots,C_L)$ only. Combining this with~\eqref{eq:Qpjjcjonel}, we conclude that
 \begin{align}
     Q=Q(C_1,\ldots,C_L)
     \end{align}
     is a function of~$(C_1,\ldots,C_L)$ only as required in the definition of~$\PCliff{Q}_L$.
     
     With~Lemma~\ref{lem:technicalcommutinglemma}, as well as expressions~\eqref{eq:explicitexpressionpteleportn} and~\eqref{eq:poutPCliffQL}
     for the probabilities $p^{\teleport_n}(P_0,\ldots,P_{n-1}\mid 
 C_0',\ldots,C_{n-1}')$ and~$p^{\PCliff{Q}_L}(P\mid C_1,\ldots, C_{L},D)$, we conclude that      $\cB_{L,m,\cJ}^\mathsf{C}$ indeed produces a valid solution for the problem~$\psim{\PCliff{Q}_L}$, as claimed.
 \end{proof}

In the next section, we will use a deterministic oracle for $\psim{\PCliff{Q}_L}$ to learn stabilizers of states $I\otimes (C_L\cdots C_1)\ket{\Phi}$.

\subsection{From $\psim{\PCliff{Q}_L}$ to 
learning stabilizers of $(I\otimes (C_L\cdots C_1))\ket{\Phi}$ \label{sec:possibilistictomographyredc}}
By definition, the output~$P$ of the circuit~$\PCliff{Q}_L$ (see Fig.~\ref{fig:PCliff})
can be interpreted as outputting the result of measuring the state
\begin{align}
    (I\otimes (C_L\cdots C_1))\ket{\Phi}\ \label{eq:interestingstateproduct}
\end{align}
in a basis of maximally entangled states obtained by rotating the Bell basis by $I\otimes Q^\dagger D^\dagger$. Here the effect of the Pauli~$Q=Q(C_1,\ldots,C_L)$ is simply to possibly flip the measurement result (i.e., it can equivalently be achieved by classical post-processing). 
By choosing different inputs~$(C_1,\ldots,C_L,D)$
with $(C_1,\ldots,C_L)$ fixed and varying over~$D\in\cC$, access to a deterministic oracle~$\cO$ for the problem $\psim{\PCliff{Q}_L}$ can thus be used by an algorithm to  simulate the 
effect of measuring the state~\eqref{eq:interestingstateproduct} in different bases of stabilizer states. A protocol using such measurement results to determine the state~\eqref{eq:interestingstateproduct} (or learn some information about it) therefore allows to determine (or learn some information about) the product~$C_L\cdots C_1$. As we will discus in Section~\ref{sec:quantumadvantage}, such information can in turn be used to solve the word problem for the Clifford group on a single qubit.

In this section, we give a protocol for  
learning the stabilizers of the state~\eqref{eq:interestingstateproduct}
using a deterministic oracle~$\cO$ for ~$\psim{\PCliff{Q}_L}$. More precisely, the corresponding algorithm
learns two distinct stabilizer generators of the state~\eqref{eq:interestingstateproduct} up to signs, see Theorem~\ref{thm:mainreductiongrier} below.  In other words, it finds two distinct (non-identity) Pauli operators that have expectation~$1$ or~$-1$ on this state. Allowing for this  remaining sign ambiguity has the key advantage that the function~$Q(C_1,\ldots,C_L)$ does not need to be evaluated in the protocol.

The protocol borrows heavily from the pioneering work~\cite{grier2020interactive}. Indeed,  Theorem~\ref{thm:mainreductiongrier} below is obtained by specializing the results of~\cite{grier2020interactive} to the problem at hand.

\subsubsection{On the need for possibilistic state tomography}
We emphasize that the problem of determining stabilizers of the state~\eqref{eq:interestingstateproduct}
using an oracle~$\cO$ for $\psim{\PCliff{Q}_L}$ cannot directly be solved using standard techniques for (stabilizer) quantum state tomography.  This is because of the possibilistic nature of the simulation provided by~$\cO$. Accordingly, we refer to the problem under consideration as {\em possibilistic tomography}.

To clarify the special difficulties associated with possibilistic tomography, consider a setting where several copies of an unknown stabilizer state~$\ket{\Psi}\in (\mathbb{C}^2)^{\otimes 2}$ can be generated/are available. This is the setting of standard tomography. Here a  set~$\{S_1,S_2\}$ of
  generators of the stabilizer group (up to a phase) of~$\ket{\Psi}$
  can efficiently be determined by a sequence of measurements with respect to different Pauli bases, see e.g.,~\cite{gottesmanaaronsonidentifying,montanarolearningstabilizerstates}.
  
  One of the key achievements of~\cite{grier2020interactive} is to extend this stabilizer state tomography to the possibilistic setting, where  the process of measuring a state~$\ket{\Psi}$ in a Pauli basis is replaced by an oracle~$\cO$ simulating the measurement.
  Remarkably, possibilistic simulation is sufficient, i.e., the simulated measurement results produced by~$\cO$ only need to lie in the support of the 
  output distribution defined by Born's rule for measurements. While this notion of  simulation has been referred to as possibilistic in~\cite{wang2021possibilistic}, the authors of~\cite{grier2020interactive} express this by saying that~$\cO$ can be  adversarial.

\subsubsection{Possibilistic state tomography from the magic square}\label{sec:posibilisticstatetomography}
The procedure of~\cite{grier2020interactive} reduces possibilistic stabilizer state tomography to possibilistic circuit simulation. Here we restate their result in a simplified and specialized  form suitable for our purposes:
\begin{theorem}[Adapted from~\cite{grier2020interactive}]\label{thm:mainreductiongrier}
Let $Q:\cC^L\rightarrow\cP$ be arbitrary. 
 Suppose~$\cO$ is a deterministic oracle (i.e., a function)  solving~$\psim{\PCliff{Q}_L}$. 
 There is an algorithm $\cA^{\cO}$ with access to~$\cO$ with the following properties:
 \begin{enumerate}[(i)]
 \item
 $\cA^{\cO}$ takes as input a sequence~$(C_1,\ldots,C_L)\in\cC^L$. 
 \item
 The algorithm~$\cA^{\cO}$ either reports failure or
 outputs a pair $(S_1,S_2)$ of two distinct two-qubit Pauli operators 
 of the form $S_j=P_A^j\otimes P_B^j$ with $P_A^j,P_B^j\in \cP\setminus \{I\}$, for $j=1,2$. The latter (successful) event happens with probability at least $2/3$. 
 \item\label{it:thirdpropertybelow}
 When the algorithm succeeds, each operator~$S_j$ for $j=1,2$ has expectation value~$\{\pm 1\}$ in the state~$(I\otimes (C_L\cdots C_1))\ket{\Phi}$.
 \item
  The algorithm~$\cA^{\cO}$ makes~$O(1)$ queries to~$\cO$.
  \item
  The algorithm~$\cA^{\cO}$ is realizable by an $\NC^0$-circuit which uses uniformly random bits and has oracle access to $\cO$.
 \end{enumerate} 
 \end{theorem}
 In the remainder of this section, we outline the proof of Theorem~\ref{thm:mainreductiongrier}. We first show  that the oracle~$\cO$ can be used to learn a non-stabilizer~$M$ of the state~$(I\otimes (C_L\cdots C_1))\ket{\Phi}$. Here a \emph{non-stabilizer} of a stabilizer state~$\ket{\Psi}$ is a Pauli operator~$M$ such that $\bra{\Psi} M\ket{\Psi}=0$.
  \begin{lemma}\label{lem:magicsquarelearning}
     Let~$\cO$ be a deterministic oracle for~$\psim{\PCliff{Q}_L}$. Then there is an $\NC^0$-algorithm~$\cN^{\cO}$
    which takes as input an $L$-tuple $(C_1,\ldots,C_L)\in\cC^L$,
\begin{enumerate}[(i)]
    \item makes $6$~queries to~$\cO$, and 
    \item outputs a non-stabilizer~$M$ of the state~$(I\otimes (C_L\cdots C_1))\ket{\Phi}$. \label{it:secondclaimnonstabilizerm}
    \end{enumerate}
    \end{lemma}
  The proof of Lemma~\ref{lem:magicsquarelearning}  exploits  the fact that measurements of the observables of the magic square illustrated in Fig.~\ref{fig:magicsquare} cannot be explained by a non-contextual model. 
    \begin{figure}
        \centering
        \begin{tabular}{ c c c }
	        $X\otimes X$ & $\mspace{3mu}Y\otimes Y$ & $\mspace{2mu}Z\otimes Z$\\
        	$Y\otimes Z$ & $\mspace{5mu}Z\otimes X$ & $X\otimes Y$\\
        	$Z\otimes Y$ & $X\otimes Z$ & $\mspace{4mu}Y\otimes X$
        \end{tabular}
        \caption{The matrix~$(P_{j,k})_{j,k=1}^3$ of Pauli observables forming the magic square.}
        \label{fig:magicsquare}
    \end{figure}
  
 \begin{proof}
 Each row~$j\in [3]$ of the magic square in Fig.~\ref{fig:magicsquare} consists of three (dependent) commuting Pauli observables~$(P_{j,1},P_{j,2},P_{j,3})$ which can be jointly measured, i.e., give rise to a measurement with outcomes~$(m_{j,1},m_{j,2},m_{j,3})\in \{\pm 1\}^3$ corresponding to the eigenvalues of the corresponding operators. Similarly, each column~$k\in [3]$ of the magic square defines a measurement with three (dependent) measurement outcomes~$(m'_{1,k},m'_{2,k},m'_{3,k})$
 associated with observables~$(P_{1,k},P_{2,k},P_{3,k})$. In total, the magic square thus defines six measurements, one for each row and column.
 
 Suppose we have six copies of a two-qubit stabilizer state~$\ket{\Psi}$. Using three copies to perform the measurements associated with each row gives a matrix~$m=(m_{j,k})\in\mathsf{Mat}_{3\times 3}(\{\pm 1\})$ of measurement outcomes. Using the remaining three copies to perform the measurements associated which each column similarly gives a matrix~$m'=(m'_{j,k})\in\mathsf{Mat}_{3\times 3}(\{\pm 1\})$. Because of the ``magic square relations"
 \begin{align}
    P_{j,1}P_{j,2}P_{j,3}\ &=-I\qquad\textrm{ for every }\qquad j\in [3]\\
    P_{1,k}P_{2,k}P_{3,k}&= \ I \mspace{8mu}\qquad\textrm{ for every }\qquad k\in [3]
 \end{align}
 satisfied by the observables~$\{P_{j,k}\}_{j,k}$ in the magic square, we have
 \begin{align}
    m_{j,1}m_{j,2}m_{j,3}\ &=-1\qquad\textrm{ for every }\qquad j\in [3]\label{eq:magicsquareonm}\\
     m'_{1,k}m'_{2,k}m'_{3,k}&=1\mspace{14mu}\qquad\textrm{ for every }\qquad k\in [3]\ .\label{eq:magicsquaretwom}
 \end{align}
 Because a magic square, i.e., a $\{\pm 1\}$-valued matrix satisfying both~\eqref{eq:magicsquareonm} and~\eqref{eq:magicsquaretwom} does not exist, there must be a pair~$(j,k)\in [3]^2$ such that $m_{j,k}\neq m'_{j,k}$. The corresponding Pauli operator~$P_{j,k}$ thus leads to a measurement result~$+1$ and a measurement result~$-1$, depending on the context. This means that~$P_{j,k}$ must be a non-stabilizer of~$\ket{\Psi}$. In summary, a non-stabilizer~$M$ of a two-qubit stabilizer state~$\ket{\Psi}$ can be obtained with certainty by performing Pauli measurements on six copies of~$\ket{\Psi}$.
 
 We can apply this idea to construct  an algorithm~$\cN^{\cO}$
 as specified in the lemma using an oracle~$\cO$ for~$\psim{\PCliff{Q}_L}$. Recall that on  input~$(C_1,\ldots,C_L,D)\in\cC^{L+1}$, the oracle~$\cO$
produces a Pauli operator~$P$ occurring with non-zero probability~$p^{\PCliff{Q}_L}(P\ |\ C_1,\ldots,C_L,D)$
in the output distribution of~$\PCliff{Q}_L$. We note that this output probability can be written as
\begin{align}
     p^{\PCliff{Q}_L}(P\ |\ C_1,\ldots,C_L,D)&=
     \left|\bra{\Phi_P(D)}(I\otimes (QC_L\cdots C_1))\ket{\Phi}
     \right|^2\ ,
\end{align}
where $Q=Q(C_1,\ldots,C_L)$ is a function of $C_1,\ldots,C_L$ and of no other arguments, and where 
\begin{align}
    \ket{\Phi_P(D)}:=(I\otimes D^\dagger)(P\otimes I )\ket{\Phi}\qquad\textrm{ for }\qquad P\in\cP
\end{align}
defines a rotated Bell basis~$(I\otimes D^\dagger)\cB:=\{\ket{\Phi_P(D)}\}_{P\in\cP}$ (with $\cB:=\{\ket{\Phi_P(I)}\}_{P\in\cP}$ being the standard Bell basis),  see~expression~\eqref{eq:pcliffqoutputdistributionm}. In other words, this output distribution is identical to that obtained  when measuring the state
\begin{align}
    I\otimes (Q C_L\cdots C_1)\ket{\Phi}\label{eq:statetobemeasuredrcl}
    \end{align}
    using the von Neumann measurement with 
    the orthonormal basis~$(I\otimes D^\dagger)\cB$.
    
    We argue below that by suitably choosing~$D$ and applying the von Neumann measurement with basis~$(I\otimes D^\dagger)\cB$, we can realize any row- and every column-measurement associated with the magic square. This means that query access to~$\cO$ with~$(C_1,\ldots,C_L)$ fixed and $D\in\cC$ arbitrary 
    immediately translates to the (possibilistic) simulation of
    magic square-measurements in the state~\eqref{eq:statetobemeasuredrcl}. 
    In particular,    six queries to~$\cO$
    of the form~$(C_1,\ldots,C_L,D)$ with $D\in \cC$ 
    reveal a non-stabilizer~$M$ of the state~\eqref{eq:statetobemeasuredrcl}. The claim~\eqref{it:secondclaimnonstabilizerm} follows from this using the following simple fact: if~$M$ is a non-stabilizer of a stabilizer state~$(I\otimes Q)\ket{\Psi}$ and $Q\in\cP$ is arbitrary, then $M$ is also a non-stabilizer of the state~$\ket{\Psi}$.
    
    It remains to argue that the measurement of any row or column of operators in the magic square is equivalent to a von Neumann measurement in a rotated Bell basis~$(I\otimes D^\dagger)\cB$ for a suitably chosen~$D\in \cC$. To this end, observe that for any two single-qubit Cliffords~$D_A,D_B\in\cC$ and any Pauli~$P'\in\cP$, we have 
    the identity
    \begin{align}
(D_A\otimes D_B)(P'\otimes I)\ket{\Phi}&=(P\otimes I)(D_A\otimes D_B)\ket{\Phi}
\quad\, \textrm{where }\ P:= D_A P' D_A^\dagger \label{eq:prenamingxm}\\
          &= (P\otimes I)(I\otimes D_BD_A^T )\ket{\Phi}\\
          &= (I\otimes D^\dagger )(P\otimes I)\ket{\Phi} \qquad \text{ where }\  D:=(D_BD_A^T)^\dagger\label{eq:PhiPDexpression}\\
          &=\ket{\Phi_P(D)} \ .
          \end{align}
          This is because  the Bell state~$\ket{\Phi}$ satisfies $(I\otimes L)\ket{\Phi}=(L^T\otimes I)\ket{\Phi}$ for any  operator~$L$ on~$\mathbb{C}^2$. 
Definitions of $P,D$ from Eq.~\eqref{eq:prenamingxm} and~\eqref{eq:PhiPDexpression} imply that---up to a relabeling of the measurement results as expressed by~\eqref{eq:prenamingxm}---the von Neumann measurement in the rotated Bell basis~$(I\otimes D^\dagger)\cB$ is equivalent to a von Neumann measurement in the basis~$(D_A\otimes D_B)\cB$  consisting of the rotated Bell states
\begin{align}
(D_A\otimes D_B)(P \otimes I)\ket{\Phi}\qquad\textrm{ for }\qquad P\in\cP\ .
    \end{align}
    Elements of the Bell basis~$\cB$ are stabilizer states having stabilizer generators of the form
    \begin{align}
        (\pm X_A\otimes X_B, \pm  Z_A\otimes Z_B)\ .
    \end{align}
    It follows that elements of the rotated Bell basis~$(D_A\otimes D_B)\cB$ have stabilizer generators of the form
    \begin{align}
    (\pm X_A'\otimes X'_B, \pm Z_A'\otimes Z_B')\qquad~\textrm{ where }~\qquad 
    \begin{cases}
    X_A'\hspace{-0.25cm}&=D_AX_AD_A^\dagger\\
    Z_A'\hspace{-0.25cm}&=D_AZ_AD_A^\dagger\\
    X_B'\hspace{-0.25cm}&=D_BX_BD_B^\dagger\\
    Z_B'\hspace{-0.25cm}&=D_BZ_BD_B^\dagger\ .
    \end{cases}\label{eq:xazaxbzb}  
    \end{align}
A von Neumann measurement in the rotated Bell basis~$(D_A\otimes D_B)\cB$---which, as shown above, is equivalent to a measurement in the basis~$(I\otimes D^\dagger)\cB$---reveals the basis state, that is, the pair of eigenvalues of~$X_A'\otimes X'_B$ and $Z_A'\otimes Z_B'$, respectively. 

We note that any pair $(P_1\otimes P_2,Q_1\otimes Q_2)$ of Pauli observables with $P_1,P_2,Q_1,Q_2\in\{X,Y,Z\}$ and $P_j\neq Q_j$ for $j=1,2$ can be written in the form~\eqref{eq:xazaxbzb} for suitably chosen Clifford operators~$D_A,D_B$. It is straightforward to check that every row and column  of the magic square in Fig.~\ref{fig:magicsquare} contains a pair~$(X_A'\otimes X'_B,Z_A'\otimes Z_B')$ of two-qubit operators satisfying this condition (the third operator is 
equal to~$\pm (X_A'\otimes X'_B)(Z_A'\otimes Z_B')$ and its eigenvalue can be obtained from the measurement results of the other two operators). Thus each row or column of the magic square can be measured by using a von Neumann measurement in a rotated Bell basis of the form~$(I\otimes D^\dagger)\cB$, as claimed.

The fact that the algorithm~$\cN^{\cO}$ is realizable by an $\NC^0$-circuit is obvious since it only manipulates a constant number of bits.
 \end{proof}

The algorithm~$\cN^{\cO}$ of Lemma~\ref{lem:magicsquarelearning}
only provides a single non-stabilizer. However, using a randomization procedure based on Kilian randomization~\cite{Kilian,grier2020interactive} and using~$\cN^{\cO}$ as a subroutine, it is possible to  construct an algorithm which produces a uniformly random non-stabilizer of~$(I\otimes C_L\cdots C_1)\ket{\Phi}$. Let us denote ${\cN^{\cO}}$ by $\cR$.
\begin{theorem}[Kilian-randomization-based algorithm]\label{thm:kilian}
There is a randomized $\NC^0$-algorithm $\cN_{\textrm{uniform}}^{\cR}$, which uses polynomially many uniformly random bits. Given an input $(C_1,\ldots,C_L)\in\cC^{L}$ the algorithm $\cN_{\textrm{uniform}}^{\cR}$
\begin{enumerate}[(i)]
\item
makes $1$ query to~$\cR$, and
\item
outputs a uniformly random non-stabilizer of $(I\otimes C_L\cdots C_1)\ket{\Phi}$.
\end{enumerate}
\end{theorem}

We defer the description of~$\cN_{\textrm{uniform}}^{\cR}$ and the proof of Theorem~\ref{thm:kilian} to Appendix~\ref{app:mainreductiongrier}.

An algorithm 
such as the one described in Theorem~\ref{thm:kilian} which provides a uniformly chosen non-stabilizer of a stabilizer state~$\ket{\Psi}$  can be run repeatedly to (eventually) learn the stabilizer group of~$\ket{\Psi}$. For a sufficiently large but constant number of calls to such an algorithm,
a pair~$(S_1,S_2)$ of Pauli operators with the desired property can be computed with probability at least $2/3$.

Theorem~\ref{thm:mainreductiongrier}  follows from these arguments, see Theorem~\ref{thm:learningstabilizers} in Appendix~\ref{app:mainreductiongrier} for details.

\subsection{Quantum advantage against $\NC^0$ from $\psim{\teleport}$\label{sec:quantumadvantage}}

In this section, we finish the argument for our  result, establishing a quantum advantage for the $\psim{\teleport}$ problem against (non-uniform) $\NC^0$.

In Section~\ref{sec:simteleptosimcliffpq}, we showed that a possibilistic (non-uniform) $\NC^0$ simulator for $\teleport$ circuits implies a possibilistic $\NC^0/\poly$ simulator for $\PCliff{Q}$ for some fixed function $Q$:
\begin{align}
    \left(\psim{\teleport}\in \NC^0/\poly\right)\ \quad \Rightarrow \ \quad \left(\exists Q\colon  \psim{\PCliff{Q}}\in \NC^0/\poly\right)\,. \label{eq:simtelepnc0pimpliessimcliffpqnc0p:repeat}
\end{align}
We will use the possibilistic simulation for the $\PCliff{Q}$ problem together with possibilistic state tomography, Section~\ref{sec:possibilistictomographyredc} (specialized from \cite{grier2020interactive}), to solve a certain word problem for a fixed group with success probability at least $2/3$. The word problem we consider is as hard as $\parityproblem$.

This will lead to the following implication:
\begin{align}
\solves{\cO}{\psim{\PCliff{Q}}} \  \Rightarrow
\ \left((\NC^0/\rpoly)^{\cO} \text{ solves } \parityproblem \text{ with prob. } \ge 2/3\right)\,.
\end{align}
Combining it with Eq.~\eqref{eq:simtelepnc0pimpliessimcliffpqnc0p:repeat} and using a standard derandomization argument, we obtain
\begin{align}
    \left(\psim{\teleport}\in \NC^0/\poly\right) \qquad \Rightarrow \qquad 
    \left(\parityproblem\in \NC^0/\poly\right)\ .
\end{align}
This contradicts that $\parityproblem\not\in \NC^0/\poly$, since the output of an $\NC^0/\poly$ circuit depends on a constant number of bits of an input instance. Therefore, we obtain $\psim{\teleport}\not\in \NC^0/\poly$.

We relate $\parityproblem$ and the word problem for a fixed group in Section~\ref{sec:parityandwp} and use it to complete the proof of quantum advantage from $\psim{\teleport}$ in Section~\ref{sec:psimteleportnotinnc0}.

\subsubsection{$\parityproblem$ as a word problem for the Clifford group}\label{sec:parityandwp}
The word problem for groups is defined as follows. Let $G$ be a group with identity element~$\idG\in G$. The word problem for the group $G$ is to decide for an input sequence of group elements $(g_1,\ldots,g_L)\in G^{L}$ whether the product $g_L\cdots g_1$ equals the identity element $\idG$. In our scenario we use the promise variant which promises that the product is either identity or some other chosen element $g$. We denote this problem $\WP(G,g)$. A sequence of group elements $(g_1,\ldots,g_L)$ is a yes-instance if $g_L\cdots g_1=\idG$ and a no-instance if $g_L\cdots g_1=g$.

Consider the word problem $\WP(\cC,H)$, where $\cC$ is the group of one-qubit Clifford unitaries with the standard product operation and $H$ the Hadamard unitary. We show that $\parityproblem$ reduces to this problem. Let $x\in\{0,1\}^L$ be a bit string. We will show how to use $\WP(\cC,H)$ with the input length $L$ to compute $\parityproblem$ of the input $x$. Set
\begin{align}
    g_i := H^{x_i} \qquad\text{ for all }\qquad i\in[n]  \ .\label{eq:parityHadamardencoding}
\end{align}
Since $H^2=I$, the product
\begin{align}
    g_L\cdots g_1=H^{\parityproblem(x)}=\begin{cases}
        I & \text{ if } x \text{ has even parity, and}\\
        H & \text{ if } x \text{ has odd parity}\ .\\
    \end{cases}
\end{align}
Hence we can use $\WP(\cC,H)$ to solve $\parityproblem$.

We note that this discussion may appear overly detailed to simply compute the parity. However, one may use other (Clifford) groups and corresponding word problems to address further computational tasks, along the lines of~\cite{grier2020interactive}, which we expand upon in Section~\ref{sec:WordProblems}. In particular, using for example two qubits and corresponding gate-teleportation circuits, one can construct circuits solving word problems associated with~$S_6$, related to Barrington's theorem~\cite{barrington89}. 
However, because of the use of lightcone-type arguments, naive applications of our reasoning cannot provide quantum advantage beyond $\NC^0$.

\subsubsection{$\psim{\teleport}$ cannot be in $\NC^0$}\label{sec:psimteleportnotinnc0}

In the following lemma we use a deterministic oracle for the $\psim{\PCliff{Q}}$ problem together with the possibilistic state tomography from Section~\ref{sec:possibilistictomographyredc} to solve $\parityproblem$.

\begin{lemma}\label{lem:parityfrompossibilisticlearning}
 Let $\cO$ be a deterministic oracle for the $\psim{\PCliff{Q}}$ problem. There is an $\NC^0$-circuit with oracle access to $\cO$ which uses uniformly random bits, makes $O(1)$ oracle queries and solves the $\parityproblem$ problem with success probability at least $2/3$.
\end{lemma}
\begin{proof} 
Let $x\in\{0,1\}^L$ be an input instance of $\parityproblem_L$. Set, as in Eq.~\eqref{eq:parityHadamardencoding}, $(C_1,\ldots, C_L)$ where $C_i:=H^{x_i}$. Again $C_L\cdots C_1=H^{\parityproblem(x)}$.

We will use the learning procedure with the oracle for the $\psim{\PCliff{Q}_L}$ problem established in 
Theorem~\ref{thm:mainreductiongrier} to learn the stabilizer group (up to the sign factor) of the state
\begin{align}
    (I\otimes C_L\cdots C_1)\ket{\Phi}=\left(I\otimes H^{\parityproblem(x)}\right)\ket{\Phi} \ .
\end{align}
This theorem tells us that there is a randomized $\NC^0$ algorithm that makes $O(1)$ queries to the oracle $\cO$ and  either returns a pair $(S_1,S_2)\in (\cP\otimes \cP)^2$ of two distinct (non-identity) Pauli operators or reports failure. If the algorithm returns $(S_1,S_2)$,
which happens with probability at least $2/3$, each $S_i$ has expectation value $\pm 1$ in the state $\left(I\otimes H^{\parityproblem(x)}\right)\ket{\Phi}$. Let us list the possible outcomes for the two options, i.e., when $x$ has even respectively odd parity:
\begin{enumerate}[(i)]
    \item If $x$ has even parity the algorithm learns the stabilizer group of the state $\left(I\otimes H^{\parityproblem(x)}\right)\ket{\Phi}=\ket{\Phi}$. It is straightforward to check that the valid outputs $(S_1,S_2)$ of the learning algorithm are two distinct elements of the set $S_{\text{even}}:=\{X\otimes X, Z\otimes Z, Y\otimes Y\}$.
    \item If $x$ has odd parity the algorithm learns the stabilizers of the state $(I\otimes H)\ket{\Phi}$. The output $(S_1, S_2)$ are two distinct elements of the set $S_{\text{odd}}:=\{X\otimes Z, Z\otimes X, Y\otimes Y\}$.
\end{enumerate}
Since $|S_{\text{even}}\cap S_{\text{odd}}|=1$, we can determine the parity of $x$ from the output $(S_1,S_2)$. Therefore we can compute $\parityproblem(x)$ with probability at least~$2/3$. This completes the proof.
\end{proof}

We use this result to complete our argument and show that there cannot be a (non-uniform) $\NC^0$-circuit for $\psim{\teleport}$. Since the $\teleport_n$ quantum circuit trivially solves $\psim{\teleport_n}$ this completes the proof of the main theorem, i.e., Theorem~\ref{thm:main}.

\begin{theorem}[$\psim{\teleport}\not\in\NC^0/\poly$]\label{thm:teleportquantumadvantage} The possibilistic simulation of the $\teleport$ circuit, i.e., the $\psim{\teleport}$ problem, is not in (non-uniform) $\NC^0$.
\end{theorem}

\begin{proof}For the sake of contradiction, assume that there is a (non-uniform) $\NC^0$-circuit $\mathsf{C}_n$ for the $\psim{\teleport_n}$ problem.

We will use $\mathsf{C}_n$ in a non-black-box way.
From Lemma~\ref{lem:simpcliffQtosimtelepnonsignaling} we have that there is an $L=\Theta(n)$ and a non-uniform $\NC^0$ algorithm $\cA$ that solves $\psim{\PCliff{Q}_L}$ for some fixed function $Q\colon \cC^L\to \cP$. This algorithm can be constructed from the algorithm $\cB^{\mathsf{C}_n}_{L,m,\cJ}$ (Algorithm~\ref{alg:simcliffptosimtelepalg}) by fixing the parameters $m,\cJ$ that are specific for the $\mathsf{C}_n$ in the non-uniform way, i.e., $\cA$ is in $\NC^0/\poly$.

Lemma~\ref{lem:parityfrompossibilisticlearning} tells us
that there is a randomized (uniform) $\NC^0$-circuit~$\cD$ with oracle access to~$\cA$ which solves $\parityproblem$.
Since $\cA\in\NC^0/\poly$ and $\cD$ makes only $O(1)$~queries to~$\cA$, the algorithm~$\cD^{\cA}$ can be realized by a non-uniform randomized~$\NC^0$-circuit~$\cE$.

Using standard procedures, we derandomize the randomized circuit~$\cE$ as follows. Because~$\cE$ is a randomized $\NC^0$-circuit with a single output bit, its output is determined by
a constant number of bits of the input instance~$x$, as well as a constant number of uniformly random bits. Enumerating all possible assignments of those random bits, computing the result of the circuit~$\cE$ for each assignment, and taking the majority of the results gives the parity of~$x$ with certainty. This procedure is realized by an $\NC^0/\poly$ circuit. This shows that
$\NC^0/\poly$ can solve $\parityproblem$, contradicting the fact that  $\parityproblem\not\in\NC^0/\poly$. In conclusion,  there is no (not even a non-uniform) $\NC^0$-circuit for the $\psim{\teleport}$ problem.
\end{proof}

\subsubsection{Other word problems for Clifford modulo Pauli groups}\label{sec:WordProblems}

Here we argue that the class of problems we could solve if we had an $\NC^0$ simulator for the $\psim{\teleport}$ problem is much 
broader than just computing $\Parity$. To address more general problems, we follow the lines of~\cite{grier2020interactive}.

The main idea relies on the observation that the technique used in Section~\ref{sec:possibilistictomographyredc} to distinguish the state $\ket{\Phi}$ from $(I \otimes H) \ket{\Phi}$ via classical simulation can be used to distinguish $\ket{\Phi}$ from any state $(I\otimes C)\ket{\Phi}$, where $C\in\cC/\cP\setminus \{I\}$ is a non-Pauli Clifford gate. This is because by definition of the Clifford group, any non-Pauli Clifford takes a non-stabilizer state to another, different non-stabilizer state. This allows us to solve the word problem $\WP(\cC/\cP,C)$ for any $C\in\cC/\cP\setminus\{I\}$, i.e., the problem of deciding  whether a series of single-qubit Clifford modulo Pauli gates multiplies to~$I$ or~$C$.

For example, the Clifford gate~$SH$ generates a cyclic group of order~$3$. Thus replacing the Hadamard gate~$H$ by $SH$ in our construction allows us to compute unbounded $\text{MOD}_3$ gates, i.e., to compute whether $|x|\bmod{ 3}$ of the Hamming weight of the input string~$x$ is $0$. Note that $\text{MOD}_3$ cannot be computed by constant-depth circuits with unbounded~$\Parity$-gates~\cite{Razborov1987,smolensky87} and vice versa, hence these problems are incomparable.

Combining the ability of computing~$\text{MOD}_3$ with that of computing~$\Parity$, we conclude that access to an $\NC^0$ simulator for the $\psim{\teleport}$~problem  elevates $\NC^0$-circuits to~$\NC^0[6]$, i.e., $\NC^0$-circuit with $\text{MOD}_6$ gates.

Alternatively, we can arrive at similar conclusions using Barrington's result~\cite{BarringtonWidth3}. The latter states that the
promise word problem for the solvable group~$S_3$ characterizes width-3 permutation branching  programs. The promise word problem for~$S_3$ can be treated in our setting because 
$S_3$ is isomorphic to $\cC/\cP$. The isomorphism is given by 
\begin{center}
\begin{tabular}{lll}
$I\mapsto 1$\ ,&
$H\mapsto (13)$\ ,&
$S\mapsto (12)$\ ,\\
$HS\mapsto (123)$\ ,&
$SHS\mapsto  (23$$)$\ ,&
$(HS)^2\mapsto (132)$\ .
\end{tabular}
\end{center}
\noindent

\paragraph{Two-qubit Clifford gate teleportation circuits}
A similar construction can be obtained with two-qubit (instead of one-qubit) gate teleportation. Notice that this circuit can be implemented on a ring of pairs of qubits. Assuming an $\NC^0$ possibilistic simulator for this problem and following the lightcone argument, we can use it to solve a possibilistic simulation of the intermediate problem depicted in Figure~\ref{fig:2PCliff}. We can ask what problems such a simulator allows us to solve.

\begin{figure}[!ht]
		\centering
		\includegraphics[scale=1.17]{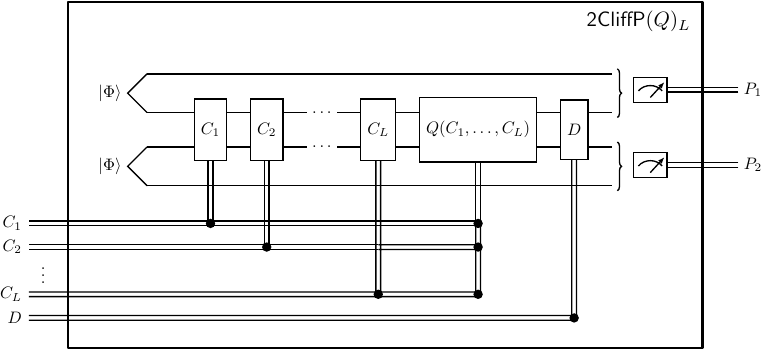}
		\caption{The Pauli-corrected two-qubit Clifford circuit $\twoPCliff{Q}_L$ with correction function $Q:\cC_2^L\rightarrow\cP_2$. It takes an input $(C_1,\ldots,C_L,D)\in\cC_2^{L+1}$ and outputs two single-qubit Pauli operators~$P_1, P_2\in \cP$ obtained from the two Bell measurements.
        \label{fig:2PCliff}
		}
\end{figure}

Let us consider the word problem $\WP(\cC_2/\cP_2,H\otimes H)$, i.e., deciding whether a series of two-qubit Clifford modulo Pauli gates multiplies to $I\otimes I$ or $H\otimes H$, as used also in Ref.~\cite{grier2020interactive}. If we have an $\NC^0$ simulator for $\psim{\twoPCliff{Q}}$ we can use it to do possibilistic tomography in a similar way to what we did in Section~\ref{sec:posibilisticstatetomography} on each of the two Bell states.

Notice that the two-qubit Clifford group modulo Pauli operators is isomorphic to $S_6$ (see Ref.~\cite{grier2020interactive}).
Since $S_6$ is a non-solvable group, the associated word problem is  complete for the complexity class $\NC^1$, i.e., for the complexity class of bounded fan-in logarithmic-depth polynomial-size circuits, according to Barrington's theorem~\cite{barrington89}. A constant-depth classical simulator for the two-qubit gate-teleportation problem could therefore be used to solve any problem in $\NC^1$.

Note, however, that none of these more general computational problems can be used to obtain a quantum advantage beyond~$\NC^0$ when only our techniques are applied. The lightcone argument we rely on constitutes a bottleneck for such a separation. The existence of a quantum advantage beyond $\NC^0$ for $1D$-local quantum circuits therefore remains an open problem.

\appendix
\section{$\psim{\PCliff{Q}}$ and possibilistic state tomography\label{app:mainreductiongrier}}
In this appendix we finish the proof of Theorem~\ref{thm:mainreductiongrier} from Section~\ref{sec:posibilisticstatetomography}. It follows by specializing the arguments given in~\cite{grier2020interactive}.

For convenience, let us define the non-stabilizer problem. Given a sequence of Cliffords $(C_1,\ldots, C_L)\in \cC^L$ as the input the \emph{non-stabilizer problem} is to output $P_A\otimes P_B\in \cP^2$ such that $\bra{\Phi} P_A\otimes P_B \ket{\Phi}=0$, i.e, such that $P_A\otimes P_B$ is a non-stabilizer (up to $\pm$ sign) of the stabilizer state $(I\otimes C_L\cdots C_1)\ket{\Phi}$.

The remainder of the arguments of~\cite{grier2020interactive} shows that an oracle~$\cO$ for the non-stabilizer problem can be used to produce a uniformly random non-stabilizer of~$(I\otimes C_L\cdots C_1)\ket{\Phi}$. We can then use this algorithm to produce generators of the stabilizer group of this state.

The algorithm for a uniformly random non-stabilizer relies on the following:
\begin{lemma}\label{lem:nonstabilizers}
Let $\ket{\Psi}$ be a stabilizer state and $M$ a non-stabilizer of~$\ket{\Psi}$.
  If $V$ is a Clifford unitary chosen uniformly at random among the set of Cliffords satisfying~$V\ket{\Psi}=\ket{\Psi}$, then $VMV^\dagger$ is a uniformly chosen non-stabilizer of~$\ket{\Psi}$.
\end{lemma} 
For a proof see~\cite[Proof of Theorem 24]{grier2020interactive}. 

Consider the Algorithm~\ref{alg:kilianrandomization} below. It builds on Kilian's randomization~\cite{Kilian} and satisfies the properties proven in Lemma~\ref{lem:kilianlemma}, see~\cite[Theorem~24]{grier2020interactive}.

\begin{algorithm}
\caption{Kilian-randomized function~$\cM^{\cO}$. Here $\cO$ is an oracle outputting a non-stabilizer of~$(I\otimes C_L\cdots C_1)\ket{\Phi}$ on input~$(C_1,\ldots,C_L)\in\cC^{L}$.
\label{alg:kilianrandomization}}
\begin{algorithmic}[1]
\Function{$\cM^{\cO}$}{$C_1,\ldots,C_L$}
\State $D_1,\ldots,D_L\gets$ uniformly random Clifford unitaries
\State $C_1'\gets D_1C_1$
\For{$j=2,\ldots,L$}\label{it:stepfouralg}
   \State $C_j'\gets D_jC_jD_{j-1}^{-1}$\label{it:stepfivealg}
\EndFor
\State $M\gets \cO(C_1',\ldots,C_{L}')$\label{eq:acomputationalg}
\State \textbf{return} 
$(I\otimes D_L)^\dagger M (I\otimes D_L)$
\EndFunction
\end{algorithmic}
\end{algorithm}

\begin{lemma}\label{lem:kilianlemma}
Let $\cO$ be a deterministic oracle for the non-stabilizer problem. The Algorithm~\ref{alg:kilianrandomization} satisfies:
\label{lem:paprimesupport}
\begin{enumerate}[(i)]
\item\label{it:firstclaimstabilizerx}
On input~$(C_1,\ldots,C_L)$, the algorithm~$\cM^\cO$ makes one query to the oracle $\cO$ and returns a non-stabilizer of~$(I\otimes \cfinal)\ket{\Phi}$.
\item\label{it:secondclaimstabilizer}
There is a random variable~$\widetilde{M}$ on~$\cP$ such that
for any sequence~$C_1,\ldots,C_L$ of Cliffords, the output~$\cM^\cO(C_1,\ldots,C_L)$ satisfies
\begin{align} 
    \cM^\cO(C_1,\ldots,C_L)&=\left(I\otimes\cfinal\right) \widetilde{M}\left(I\otimes\cfinal\right)^\dagger
\end{align}
with certainty. In particular, the dependence  of the output of~$\cM^\cO$ on the input~$(C_1,\ldots,C_L)$ is only via the product~$\cfinal$.
\item\label{it:distributionxproperty}
The support of the distribution~$P_{\widetilde{M}}$ (i.e., the distribution of $\widetilde{M}$) is contained in the set of non-stabilizers of the state~$\ket{\Phi}$.
\end{enumerate}
\end{lemma}
\begin{proof}
Claim~\eqref{it:firstclaimstabilizerx} follows immediately from
the definition of~$\cM^\cO$: the operator~$M$ obtained in Step~\eqref{eq:acomputationalg} of the Algorithm~\ref{alg:kilianrandomization} is a non-stabilizer of
$I\otimes(C'_L\cdots C_1')\ket{\Phi}=(I\otimes D_L(\cfinal))\ket{\Phi}$, hence $(I\otimes D_L)^\dagger M(I\otimes D_L)$ is a non-stabilizer of~$(I\otimes \cfinal)\ket{\Phi}$.

For the proof of~\eqref{it:secondclaimstabilizer}, consider the following procedure, given $(C_1,\ldots,C_L)$: pick $D_L$ uniformly at random from the Clifford group, then pick $(C_1',\ldots,C_L')$ uniformly from the set of sequences satisfying
\begin{align}
    C_L'\cdots C_1'&=D_L \cfinal\ .\label{eq:dlcfinalexpr}
\end{align}
Then set
\begin{align}
    M=\cO(C_1',\ldots,C_L')
\end{align}
and
\begin{align}
    \widetilde{M}&=(I\otimes(C_L'\cdots C_1'))^\dagger M (I\otimes(C_L'\cdots C_1'))\ .\label{eq:tildeMprodexpr}
\end{align}
Observe first that the distribution $P_{\widetilde{M}}$ is
a function of~$\cfinal$ only and does not depend on the choice of~$(C_1,\ldots,C_L)$ otherwise.
Indeed, the sequence~$(C_1,\ldots,C_L)$ only determines~$\cfinal$,
implying that the product~$D_L\cfinal$ in~\eqref{eq:dlcfinalexpr}
has a distribution which is a function of the product~$\cfinal$ only.  Because the distribution of~$(C_1',\ldots,C_L')$ is fixed by~$D_L\cfinal$ by definition, it follows that the distribution of~$M$ is determined by the product~$\cfinal$ (with no additional dependence on~$(C_1,\ldots,C_L)$). It follows similarly from~\eqref{eq:tildeMprodexpr} that
the distribution of~$\widetilde{M}$ has this property.

By definition, we have 
\begin{align}
    \left(I\otimes \cfinal\right) \!\widetilde{M}\! \left(\!I\otimes\cfinal\right)^\dagger\!\! &=
    \left(\!I\otimes\! \left(\cfinal\right)\!(C'_L\cdots C_1')^\dagger\right)\! M\! \left((I\otimes\!(C'_L\cdots C'_1)\cfinal\right)^\dagger\\
    &=    \left(I\otimes \cfinal \left(D_L\cfinal\right)^\dagger\right)\!M \!\left(I\otimes\left(D_L\cfinal\right)\cfinal\right)^\dagger\\
    &=(I\otimes D_L)^\dagger M (I\otimes D_L)\\
    &=(I\otimes D_L)^\dagger \cO(C_1',\ldots,C_L') (I\otimes D_L)\label{eq:outputdlacprm}
\end{align}
where we used~\eqref{eq:dlcfinalexpr} in the second step. Observe that the marginal distributions $P_{C_1'\cdots C_L'D_L}$ are identical for both algorithms~$\cM$ and the procedure used in this proof: this is because steps~\eqref{it:stepfouralg} and~\eqref{it:stepfivealg} amount to Kilian randomization \cite{Kilian}, producing a uniformly random sequence~$(C_1',\ldots,C'_L)$ of group elements with a fixed  product as in~\eqref{eq:dlcfinalexpr}. It follows that~\eqref{eq:outputdlacprm} is indeed identical to the output of~$\cM$ as claimed.

Finally, the claim~\eqref{it:distributionxproperty} is an immediate consequence of~\eqref{it:firstclaimstabilizerx} and~\eqref{it:secondclaimstabilizer}. 
\end{proof}

\begin{algorithm}
\caption{Randomization function~$\cM'$.\label{alg:randomizedsecond}}
\begin{algorithmic}[1]
\Function{$\cM'$}{$C_1,\ldots,C_L$}
\State $V\gets$ a random Clifford satisfying $V\ket{\Phi}=\ket{\Phi}$.
\State $C_1'\gets C_1V$
\For{$j=2,\ldots,L$}
   \State $C_j'\gets C_j$
\EndFor
\State $M\gets \cM^\cO(C_1',\ldots,C_{L}')$\label{eq:apcomputationalg}
\State \textbf{return} 
$M$
\EndFunction
\end{algorithmic}
\end{algorithm}

Let $\cM$ be as specified in Algorithm~\ref{alg:kilianrandomization} with an oracle access to the oracle~$\cO$ that solves the non-stabilizer problem. Consider the algorithm~$\cM'$ specified in Algorithm~\ref{alg:randomizedsecond}. Then we have the following (see~\cite[Theorem 24]{grier2020interactive}):
\begin{lemma}
On input~$(C_1,\ldots,C_L)$, the algorithm~$\cM'$ returns a uniformly distributed non-stabilizer of~$\left(I\otimes\cfinal\right)\ket{\Phi}$.
\end{lemma}
\begin{proof}
Let $\widetilde{M}$ denote the Pauli-valued random variable from Lemma~\ref{lem:paprimesupport}. Let $V$ be a uniformly distributed one-qubit Clifford unitary subject to the constraint~$V\ket{\Phi}=\ket{\Phi}$. Then we have 
\begin{align}
P_{\cM'(C_1,\ldots,C_L)}&=P_{\cM^{\cO}(C_1V,C_2,\ldots,C_L)}\\
&=P_{\cfinal V\widetilde{M}V^\dagger \cfinal^\dagger}
\end{align}
by definition of the algorithm~$\cM'$ and Lemma~\ref{lem:paprimesupport}~\eqref{it:secondclaimstabilizer}.
Since $\widetilde{M}$ is supported on the set of non-stabilizers of the state~$\ket{\Phi}$ (see 
Lemma~\ref{lem:paprimesupport}~\eqref{it:distributionxproperty}), it follows from Lemma~\ref{lem:nonstabilizers} that $V\widetilde{M}V^\dagger$ is a random non-stabilizer of~$\ket{\Phi}$ with uniform distribution. The claim follows from this.
\end{proof}

Finally, we combine a deterministic oracle for the non-stabilizer problem and the randomization procedure that we specified in Algorithms~\ref{alg:kilianrandomization} and~\ref{alg:randomizedsecond} to an algorithm learning the stabilizer group of the state $(I\otimes C_L\cdots C_1)\ket{\Phi}$ (up to a sign factor).

\begin{theorem}\label{thm:learningstabilizers}
Let~$\cO$ be a deterministic oracle for the non-stabilizer problem. 
There is an $\NC^0$ algorithm with an oracle access to~$\cO$ that uses uniformly random bits, and on input~$(C_1,\ldots,C_L)$ 
\begin{enumerate}[(i)]
\item
 makes $O(1)$ oracle calls to~$\cO$, 
\item either
outputs two distinct non-identity Pauli operators~$S_1,S_2$ such that 
 \begin{align}
     \bra{\Phi}\left(I\otimes (C_L\cdots C_1)\right)^\dagger S_j \left(I\otimes (C_L\cdots C_1)\right)\ket{\Phi}\in \{\pm 1\}\ \quad\textrm{ for }\ \quad j=1,2\ ,
 \end{align}
 or  reports a failure, and
 \item
 succeeds with probability at least~$2/3$.
\end{enumerate}
\end{theorem}
 \begin{proof}
 Let $\ket{\Psi}$ be a stabilizer state. A Pauli operator~$M$ has expectation value $\bra{\Psi}M\ket{\Psi}\in \{\pm 1\}$ if and only if~$M$ commutes with each stabilizer of~$\ket{\Psi}$. If this is not the case, then $M$ is a non-stabilizer. It follows that a two-qubit stabilizer state has exactly~$12$ non-stabilizers of the form~$P_A\otimes P_B$ with $(P_A,P_B)\in \cP^2$. The remaining~$4$ operators constitute (up to~$\pm$-signs) the stabilizer group of~$\ket{\Psi}$. In particular, given the $12$~non-stabilizers of~$\ket{\Psi}$, it is straightforward to output two distinct non-identity Pauli operators~$S_1=P_A\otimes P_B$, 
$S_2=P_A'\otimes P_B'$ such that $\bra{\Psi}S_j\ket{\Psi}\in \{\pm 1\}$ for $j=1,2$.

We apply this to the state~$\ket{\Psi}=(I\otimes C_L\cdots C_1)\ket{\Phi}$.  Let $\cM'$ be the algorithm defined in Algorithm~\ref{alg:randomizedsecond}, using the function ~$\cM^{\cO}$ given in Algorithm~\ref{alg:kilianrandomization}. Notice that this algorithm returns a uniformly random non-stabilizer of $\ket{\Psi}$ given the input $C_L\cdots C_1$. Then, there is an algorithm $\cM''$ that uses $O(1)$~queries to~$\cM'$ and determines, with probability at least~$2/3$, all~$12$ non-stabilizers of~$\ket{\Psi}$ of the form $P_A\otimes P_B$ with~$(P_A,P_B)\in \cP^2$ or it reports failure. The claim follows.
\end{proof}

\subsection*{Acknowledgments} LC and RK gratefully acknowledge support by the European Research Council under grant agreement no.~101001976 (project EQUIPTNT). XCR thanks the Swiss National Science Foundation (SNSF) for their support.

\normalem


\begin{thebibliography}{10}

\bibitem{TerhalDiVincenzo04}
Barbara~M. Terhal and David~P. DiVincenzo.
\newblock Adaptive quantum computation, constant depth quantum circuits and
  {A}rthur-{M}erlin games.
\newblock {\em Quantum Information \& Computation}, 4(2):134--145, 2004.
\newblock \href {https://doi.org/10.26421/QIC4.2-5}
  {\path{doi:10.26421/QIC4.2-5}}.

\bibitem{wang2021possibilistic}
Daochen Wang.
\newblock Possibilistic simulation of quantum circuits by classical circuits.
\newblock {\em Phys. Rev. A}, 106:062430, Dec 2022.
\newblock \href {https://doi.org/10.1103/PhysRevA.106.062430}
  {\path{doi:10.1103/PhysRevA.106.062430}}.

\bibitem{gottesmanchuang}
Daniel Gottesman and Isaac Chuang.
\newblock Demonstrating the viability of universal quantum computation using
  teleportation and single-qubit operations.
\newblock {\em Nature}, 402(6760):390 -- 393, 1999.
\newblock \href {https://doi.org/10.1038/46503} {\path{doi:10.1038/46503}}.

\bibitem{bravyigossetkoenig}
Sergey Bravyi, David Gosset, and Robert König.
\newblock {Quantum advantage with shallow circuits}.
\newblock {\em Science}, 362(6412):308--311, oct 2018.
\newblock \href {https://doi.org/10.1126/science.aar3106}
  {\path{doi:10.1126/science.aar3106}}.

\bibitem{BravyiGossetKoenigTomamichel}
Sergey Bravyi, David Gosset, Robert K{\"o}nig, and Marco Tomamichel.
\newblock Quantum advantage with noisy shallow circuits.
\newblock {\em Nature Physics}, 16(10):1040--1045, 2020.
\newblock \href {https://doi.org/10.1038/s41567-020-0948-z}
  {\path{doi:10.1038/s41567-020-0948-z}}.

\bibitem{AaronsonCh15}
Scott Aaronson.
\newblock {Quantum Computing, Postselection, and Probabilistic
  Polynomial-Time}, Dec 2004.
\newblock \href {https://doi.org/10.48550/arXiv.quant-ph/0412187}
  {\path{doi:10.48550/arXiv.quant-ph/0412187}}.

\bibitem{fennergreenhomerzhang05}
Stephen Fenner, Frederic Green, Steven Homer, and Yong Zhang.
\newblock {Bounds on the Power of Constant-Depth Quantum Circuits}.
\newblock In {\em Fundamentals of Computation Theory}, pages 44--55, Berlin,
  Heidelberg, 2005. Springer Berlin Heidelberg.
\newblock \href {https://doi.org/10.1007/11537311_5}
  {\path{doi:10.1007/11537311_5}}.

\bibitem{VanDenNesClassicalSimulation}
Maarten Van~{D}en {N}es.
\newblock Classical simulation of quantum computation, the
  {G}ottesman{-}{K}nill theorem, and slightly beyond.
\newblock {\em Quantum Information \& Computation}, 10(3):258--271, 2010.
\newblock \href {https://doi.org/10.26421/QIC10.3-4-6}
  {\path{doi:10.26421/QIC10.3-4-6}}.

\bibitem{bremnerjozsashepher11}
Michael~J. Bremner, Richard Jozsa, and Dan~J. Shepherd.
\newblock {Classical simulation of commuting quantum computations implies
  collapse of the polynomial hierarchy}.
\newblock {\em {Proceedings of the Royal Society A: Mathematical, Physical and
  Engineering Sciences}}, 467(2126):459--472, 2011.
\newblock \href {https://doi.org/10.1098/rspa.2010.0301}
  {\path{doi:10.1098/rspa.2010.0301}}.

\bibitem{Farhiharrow2016}
Edward Farhi and Aram~Wettroth Harrow.
\newblock {Quantum Supremacy through the Quantum Approximate Optimization
  Algorithm}.
\newblock Technical report: MIT/CTP-4771, 2016.
\newblock \href {https://doi.org/10.48550/arXiv.1602.07674}
  {\path{doi:10.48550/arXiv.1602.07674}}.

\bibitem{aaronsonarkhipovboson11}
Scott Aaronson and Alex Arkhipov.
\newblock {The Computational Complexity of Linear Optics}.
\newblock In {\em Proceedings of the Forty-Third Annual ACM Symposium on Theory
  of Computing}, STOC '11, page 333–342, New York, NY, USA, 2011. Association
  for Computing Machinery.
\newblock \href {https://doi.org/10.1145/1993636.1993682}
  {\path{doi:10.1145/1993636.1993682}}.

\bibitem{Grier_2022}
Daniel Grier and Luke Schaeffer.
\newblock The classification of {C}lifford gates over qubits.
\newblock {\em Quantum}, 6:734, Jun 2022.
\newblock \href {https://doi.org/10.22331/q-2022-06-13-734}
  {\path{doi:10.22331/q-2022-06-13-734}}.

\bibitem{grier2020interactive}
Daniel Grier and Luke Schaeffer.
\newblock {Interactive shallow Clifford circuits: Quantum advantage against
  NC$^1$ and beyond}.
\newblock In {\em Proceedings of the 52nd Annual ACM SIGACT Symposium on Theory
  of Computing}, pages 875--888, 2020.
\newblock \href {https://doi.org/10.1145/3357713.3384332}
  {\path{doi:10.1145/3357713.3384332}}.

\bibitem{WattsKothariSchaefferTalAC0}
Adam~Bene Watts, Robin Kothari, Luke Schaeffer, and Avishay Tal.
\newblock Exponential separation between shallow quantum circuits and unbounded
  fan-in shallow classical circuits.
\newblock In {\em Proceedings of the 51st Annual ACM SIGACT Symposium on Theory
  of Computing}, STOC 2019, page 515–526, New York, NY, USA, 2019.
  Association for Computing Machinery.
\newblock \href {https://doi.org/10.1145/3313276.3316404}
  {\path{doi:10.1145/3313276.3316404}}.

\bibitem{CHSH}
John~F. Clauser, Michael~A. Horne, Abner Shimony, and Richard~A. Holt.
\newblock Proposed experiment to test local hidden-variable theories.
\newblock {\em Physical {R}eview {L}etters}, 23(15):880, 1969.
\newblock \href {https://doi.org/10.1103/PhysRevLett.23.880}
  {\path{doi:10.1103/PhysRevLett.23.880}}.

\bibitem{BeneWattsParham2023unconditional}
Adam~Bene Watts and Natalie Parham.
\newblock {Unconditional Quantum Advantage for Sampling with Shallow Circuits},
  2023.
\newblock \href {https://doi.org/10.48550/arXiv.2301.00995}
  {\path{doi:10.48550/arXiv.2301.00995}}.

\bibitem{renou2019genuine}
Marc-Olivier Renou, Elisa B{\"a}umer, Sadra Boreiri, Nicolas Brunner, Nicolas
  Gisin, and Salman Beigi.
\newblock Genuine quantum nonlocality in the triangle network.
\newblock {\em Physical {R}eview {L}etters}, 123(14):140401, 2019.
\newblock \href {https://doi.org/10.1103/PhysRevLett.123.140401}
  {\path{doi:10.1103/PhysRevLett.123.140401}}.

\bibitem{gottesmanaaronsonidentifying}
Scott Aaronson and Daniel Gottesman.
\newblock {Identifying Stabilizer States}, 2008.
\newblock See \url{http://pirsa.org/08080052/}.

\bibitem{montanarolearningstabilizerstates}
Ashley Montanaro.
\newblock {Learning stabilizer states by Bell sampling}, Jul 2017.
\newblock \href {https://doi.org/10.48550/arXiv.1707.04012}
  {\path{doi:10.48550/arXiv.1707.04012}}.

\bibitem{Kilian}
Joe Kilian.
\newblock Founding cryptography on oblivious transfer.
\newblock In {\em Proceedings of the Twentieth Annual ACM Symposium on Theory
  of Computing}, STOC '88, page 20–31, New York, NY, USA, 1988. Association
  for Computing Machinery.
\newblock \href {https://doi.org/10.1145/62212.62215}
  {\path{doi:10.1145/62212.62215}}.

\bibitem{barrington89}
David~A. Barrington.
\newblock {Bounded-width polynomial-size branching programs recognize exactly
  those languages in $NC^1$}.
\newblock {\em Journal of Computer and System Sciences}, 38(1):150--164, Feb
  1989.
\newblock \href {https://doi.org/10.1016/0022-0000(89)90037-8}
  {\path{doi:10.1016/0022-0000(89)90037-8}}.

\bibitem{Razborov1987}
Alexander~A. Razborov.
\newblock {Lower bounds on the size of bounded depth circuits over a complete
  basis with logical addition}.
\newblock {\em Mathematical notes of the Academy of Sciences of the USSR},
  41:333--338, 1987.
\newblock \href {https://doi.org/10.1007/BF01137685}
  {\path{doi:10.1007/BF01137685}}.

\bibitem{smolensky87}
Roman Smolensky.
\newblock {Algebraic Methods in the Theory of Lower Bounds for Boolean Circuit
  Complexity}.
\newblock In {\em Proceedings of the Nineteenth Annual ACM Symposium on Theory
  of Computing}, STOC '87, page 77–82, New York, NY, USA, 1987. Association
  for Computing Machinery.
\newblock \href {https://doi.org/10.1145/28395.28404}
  {\path{doi:10.1145/28395.28404}}.

\bibitem{BarringtonWidth3}
David~A. Barrington.
\newblock Width-3 permutation branching programs, technical memorandum tm-293.
\newblock Technical report, M.I.T. Laboratory for Computer Science, Dec 1985.
\newblock URL: \url{https://hdl.handle.net/1721.1/149102}.

\end{thebibliography}
\end{document}